\documentclass[11pt]{article}
\usepackage[margin=1in]{geometry}

\usepackage{graphicx}
\usepackage{tikz}
\usepackage{circuitikz}
\usetikzlibrary{arrows.meta, positioning, quotes}
\usepackage[utf8]{inputenc}
\usepackage{amsmath}
\usepackage{amsfonts}
\usepackage{amsthm}
\usepackage{caption}
\usepackage{graphicx}
\usepackage{calrsfs}
\usepackage{dirtytalk}
\usepackage{xcolor}
\DeclareMathAlphabet{\pazocal}{OMS}{zplm}{m}{n}
\DeclareMathAlphabet{\mathcal}{OMS}{cmsy}{m}{n}
\usepackage{enumerate}   
\usepackage{bbm}
\usepackage[hyperfootnotes=false]{hyperref}
\usepackage{cleveref} 
\usepackage{natbib}
\usepackage{comment}
\usepackage{algorithm}
\usepackage{algpseudocode}

\theoremstyle{definition}
\theoremstyle{plain}
\newtheorem{definition}{Definition}
\newtheorem{theorem}{Theorem}[section]

\newtheorem{lemma}[theorem]{Lemma}
\newtheorem{fact}[theorem]{Fact}
\newtheorem{corollary}[theorem]{Corollary}

\AddToHook{env/proposition/begin}{\crefalias{theorem}{proposition}}
\AddToHook{env/claim/begin}{\crefalias{theorem}{claim}}
\AddToHook{env/lemma/begin}{\crefalias{theorem}{lemma}}
\AddToHook{env/fact/begin}{\crefalias{theorem}{fact}}
\AddToHook{env/corollary/begin}{\crefalias{theorem}{corollary}}
\AddToHook{cmd/appendix/before}{\crefalias{section}{appendix}}
\makeatother

\newcommand{\inst}{\mathbbm{I}}
\newcommand{\feas}{\mathcal{F}}
\newcommand{\mc}{\mathcal{M}}
\newcommand{\comgap}{\mathrm{ComGap}}
\newcommand{\frucorgap}{\mathrm{FruCorGap}}
\newcommand{\corgap}{\mathrm{CorGap}}
\newcommand{\frugap}{\mathrm{FruGap}}
\newcommand{\dist}{\mathcal{D}}

\newcommand{\alg}{\mathrm{ALG}}
\newcommand{\polytope}{\mathcal{P}}
\newcommand{\util}{\mathrm{Util}}

\newcommand{\val}{\mathrm{Val}}

\newcommand{\R}{\mathbb{R}}

\newcommand{\argmax}{\operatorname{argmax}}
\newcommand{\eps}{\epsilon}
\newcommand{\D}{\mathcal{D}}

\newcommand{\exanteopt}{\mathrm{ExAnteOpt}}

\newcommand{\opt}{\mathrm{Opt}}

\newcommand{\comms}{\mathcal{P}}
\newcommand{\comset}{\mathcal{C}}

\newcommand{\exante}{\mathrm{ExAnte}}
\newcommand{\expost}{\mathrm{ExPost}}

\newcommand{\freeorder}{\mathrm{FreeOrder}}
\newcommand{\semi}{\mathrm{OSSO}}
\newcommand{\trajset}{\Gamma}
\newcommand{\traj}{\gamma}

\newcommand{\mset}{\mathbb{M}}

\newcommand{\pol}{\mathrm{A}}
\newcommand{\com}{\pi}
\newcommand{\so}{\mathrm{B}}

\newcommand{\expect}[2]{\operatorname{E}_{#1}\left[#2\right]}
\newcommand{\expectt}[1]{\operatorname{E}\left[#1\right]}
\newcommand{\prob}[2]{\operatorname{Pr}_{#1}\left[#2\right]}
\newcommand{\probb}[1]{\operatorname{Pr}\left[#1\right]}

\newcommand{\disti}{D}

\renewcommand{\algorithmiccomment}[1]{{\small\color{gray}\bgroup\hfill$\vartriangleright$
#1\egroup}}

\title{Commitment Gap via Correlation Gap\footnote{This work was funded in part by NSF award CCF-2217069.}}
\author{
Shuchi Chawla\footnote{University of Texas at Austin} \\ {\tt shuchi@cs.utexas.edu} 
\and 
Dimitris Christou$^\dagger$  \\ {\tt christou@cs.utexas.edu} 
\and
Trung Dang$^\dagger$ \\ {\tt dddtrung@cs.utexas.edu}
}
\date{}

\begin{document}

\maketitle
\vspace{-15pt}
\begin{abstract}

Selection problems with costly information, dating back to Weitzman's Pandora's Box problem, have received much attention recently. We study the general model of Costly Information Combinatorial Selection (CICS) that was recently introduced by~\cite{CCHS24} and~\cite{BLW25}. In this problem, a decision maker needs to select a feasible subset of stochastic variables, and can only learn information about their values through a series of costly steps, modeled by a Markov decision process. The algorithmic objective is to maximize the total value of the selection \textit{minus} the cost of information acquisition. However, determining the optimal algorithm is known to be a computationally challenging problem.

To address this challenge, previous approaches have turned to approximation algorithms by considering a restricted class of \textit{committing policies} that simplify the decision-making aspects of the problem and allow for efficient optimization. This motivates the question of bounding the \textit{commitment gap}, measuring the worst case ratio in the performance of the optimal committing policy and the overall optimal. In this work, we obtain improved bounds on the commitment gap of CICS through a reduction to a simpler problem of Bayesian Combinatorial Selection where information is free. By establishing a close relationship between these problems, we are able to relate the commitment gap of CICS to ex ante free-order prophet inequalities. As a consequence, we obtain improved efficient approximations  for arbitrary instances of the CICS under various feasibility constraints.


\end{abstract}

\thispagestyle{empty}
\addtocounter{page}{-1}

\newpage

\section{Introduction}


Recently, there has been a surge of research on stochastic online decision-making models. These models vary in how algorithms access information and in the constraints imposed by the online nature of the process. The subtle distinctions between models often call for algorithmic techniques tailored to their specific characteristics. In this paper, we reveal that several models of online selection are intimately related. We uncover tight connections that allow techniques developed in one setting to be transferred to another, leading to improved bounds on their competitive ratios.

Our primary focus is the Costly Information Combinatorial Selection (CICS) problem, introduced by \cite{CCHS24} and \cite{BLW25}. In CICS, the algorithm is given a ground set of $n$ elements and seeks to select a feasible subset of maximum total value.\footnote{We focus primarily on the maximization version of the problem, although some of our results apply to the minimization version as well. See~\Cref{app:minimization} for a discussion of the minimization setting.} While the values of the elements are drawn independently from known distributions, their realizations are hidden and can only be uncovered through a sequence of costly information-gathering actions. The objective is to maximize the value of the selected subset minus the total cost of acquiring information.

A motivating example is a hiring process in which each candidate has a stochastic value $v_i$. Assessments---resume reviews, interviews, tests---provide increasingly refined information about the values but consume resources. The company aims to maximize the value of the recruited candidates minus the total cost of gathering information along the way. For instance, it may screen resumes broadly and interview only a small subset. Feasibility constraints model hiring limitations: selecting up to $k$ candidates yields a uniform matroid, while recruiting for multiple projects with candidate–project compatibility corresponds to bipartite matching.



An extensively studied special case of CICS is the {\bf Pandora’s Box} problem of \cite{W79}, in which each element is hidden in a box and the goal is to select a single item. Opening a box reveals its value at a fixed cost. This setting admits a simple optimal policy: boxes should be opened in a fixed order until a stopping condition is met, at which point the best observed item is selected. However, more general variants of CICS---where learning an element's value follows a Markov decision process (MDP)---are significantly more complex. The algorithm must determine, based on the evolving joint state of the $n$ MDPs, whether and how to probe an element (e.g., review a resume or conduct an interview), or to select (hire) one. Because the joint state space is exponentially large, optimal algorithms do not necessarily admit succinct descriptions and are hard to compute even for special cases~\citep{FLL22}. As a result, prior work has focused on obtaining approximately optimal solutions using simpler algorithmic frameworks.

One such framework is that of {\bf committing algorithms}, which fix all ``local'' probing actions in advance, before acquiring any information. Specifically, for each element 
$i$ and each possible state in its MDP, a committing algorithm specifies the next action to take if the algorithm chooses to probe that element. This commitment simplifies the global decision-making process: once local policies are fixed, the optimal interleaving of these decisions can be efficiently computed~\citep{GJSS19, CCHS24}. A central question is the resulting loss in performance--a quantity known as the {\bf commitment gap}.

In this paper, {\bf we develop a new approach for bounding the commitment gap in CICS by relating it to the well-studied notion of the correlation gap}. Through this connection, we show that for a broad class of feasibility constraints--those that admit good greedy algorithms--the commitment gap is provably small. For example, we show that 
committing policies achieve an approximation factor of at least $1-\frac{1}{e}$ over arbitrary instances of the CICS with a matroid feasibility constraint. See~\Cref{table:res} for a full summary of our results. All of our approximations can be achieved efficiently.


\subsubsection*{Bayesian Combinatorial Selection}
Our approach relies on a reduction from CICS to a Bayesian Combinatorial Selection (BCS) problem. In BCS, we are given a ground set of $n$ stochastic elements $X_1, \cdots, X_n$, drawn from a known joint product distribution $\dist$, and a downwards closed feasibility constraint $\feas\subseteq 2^{[n]}$. The objective is to select a feasible subset of elements with maximum total value. Unlike CICS, there is no cost associated with probing elements; however, the algorithm must adhere to certain structural rules that govern when elements can be probed and selected. The BCS framework thus provides a flexible and unified foundation for stochastic selection, into which various online models can be embedded by enforcing different constraints on the probing and selection process.

As an illustrative example, consider a {\em Fixed-Order Online Selection} algorithm. Here, elements are presented sequentially in an adversarially chosen order. Upon observing the value of each element, the algorithm must make an irrevocable decision to accept or reject it. The performance gap between such an algorithm and the offline optimum---defined as  $\expect{x\sim \dist}{\max_{S\in\feas} \sum_{i\in S} x_i}$---is captured by the classical notion of a {\em Prophet Inequality}. We show that the commitment gap in CICS is closely connected to the performance of a more powerful class of algorithms within the BCS framework, which we introduce and call {\em One-Sided Semi-Online Selection} algorithms.

A one-sided semi-online selection algorithm constructs its solution by picking an element and a threshold at every step and adding the element to the solution if and only if its value is above the threshold, while maintaining feasibility. Importantly, it can visit elements repeatedly in any order, but does not learn the value of the chosen element---only whether it lies above or below the threshold. We show that one-sided semi-online selection algorithms for a suitably defined BCS instance capture exactly the performance of committing policies for CICS. This allows us to focus our attention to the properties of this class of algorithms. 



\subsubsection*{Surrogate Values and the Ex-Ante Relaxation}

The connection between CICS and BCS arises through the notion of surrogate values that were first defined by \cite{KWG16} and \cite{S17} for the Pandora's Box problem and extended to general CICS by \cite{GJSS19} and \cite{CCHS24}. Surrogate values provide a way to amortize the cost of the probing actions taken by the algorithm to future steps where the algorithm makes a selection, thereby allowing the algorithm to essentially explore for free. 

Consider the special case of CICS where each MDP is a Markov chain---that is, in any state of the process and for any element, there is only one way to probe the element. \cite{GJSS19} showed that the optimal value of an instance of CICS over Markov chains (henceforth, MC-CICS) is bounded by the ex post optimum for the BCS over the surrogate values. On the other hand, we show that any one-sided semi-online selection algorithm for BCS over the surrogate values can be converted efficiently into an algorithm for MC-CICS obtaining the same total value. A small separation between one-sided semi-online algorithms and the ex post optimum for BCS therefore implies a 
good approximation for the MC-CICS. The converse also holds: an $\alpha$-approximation for MC-CICS obtained via a certain family of algorithms  (namely, those satisfying ``promise of payment'' as defined in~\Cref{sec:cics}) can be converted into an $\alpha$-approximate one-sided semi-online algorithm for BCS.

This approach does not immediately extend to general MDPs: different commitments in an MDP correspond to different distributions over surrogate values. We instead consider the {\em ex ante relaxation} of CICS, first proposed by \cite{BLW25}. The ex ante relaxation maximizes the expected value of the selected subset subject to the feasibility constraint {\em in expectation} over the realization of values. 
A solution to the ex ante relaxation specifies for each element the probability with which it is selected; and maximizes the value obtained from the element conditioned on this selection probability. Once the selection probabilities are determined, the per-element maximization corresponds to a specific committing policy for the element's MDP independent of other elements. The value of the ex ante benchmark can then be related to an instance of BCS over surrogate values corresponding to these commitments. These connections allow us to relate the commitment gap of the CICS to the gap between one-sided semi-online algorithms and ex ante relaxations for BCS. 

\subsubsection*{Commitment Gap, Correlation Gap, and Frugality}

To bound the gap between the ex ante optimum and one-sided semi-online algorithms for BCS, it is useful to relate each of these quantities to the ex post optimum. \cite{Y11} showed that the gap between the ex post and the ex ante optima is exactly equal to the Correlation Gap of $\feas$ (defined in~\Cref{sec:bcs}), an extensively studied quantity in stochastic optimization. Complementing this, we build on an approach from \cite{S17} to show that the gap between the ex post objective and one-sided semi-online algorithms is bounded by the approximation factor achieved by greedy-style algorithms for maximizing value subject to $\feas$ for {\em deterministic} instances. We refer to this latter gap as the {\em frugality} of $\feas$. Together these observations imply that the gap between one-sided semi-online algorithms and the ex ante optimum---which we call the {\em frugal correlation gap} and which bounds the commitment gap for CICS over $\feas$---is bounded by
the correlation gap of $\feas$ times its frugality. This bound, however, can be quite loose. 


To sharpen it, we further analyze the structure of one-sided semi-online algorithms and show that their performance relative to the ex ante benchmark matches that of the more restrictive class of {\em Free-Order Online Selection} algorithms. A free-order online algorithm behaves like a one-sided semi-online algorithm but may consider each element for selection at most once. Notably, free-order online algorithms for BCS are precisely the algorithms employed in a Free-Order Prophet Inequality. We prove, in effect, that the frugal correlation gap of any feasibility constraint $\feas$ is exactly equal to the best competitive ratio achievable for the corresponding ex ante free-order prophet inequality.

Finally, we note that the correlation gap of a feasibility constraint is exactly equal to the approximation factor achieved by Contention Resolution Schemes for the same constraint~\citep{CVZ11}; the competitive ratio achieved by Online Contention Resolution Schemes, on the other hand, is equivalent to ex ante prophet inequalities \citep{LS18}. We summarize these different benchmarks, gaps, and relationships discussed above in~\Cref{fig:gaps}. These connections allow us to use existing results for the correlation gap and free-order prophet inequalities to obtain bounds on the commitment gap for CICS; we state these results in~\Cref{table:res} and discuss them in more detail later; all gaps and approximation factors are given as ratios $\le 1$.

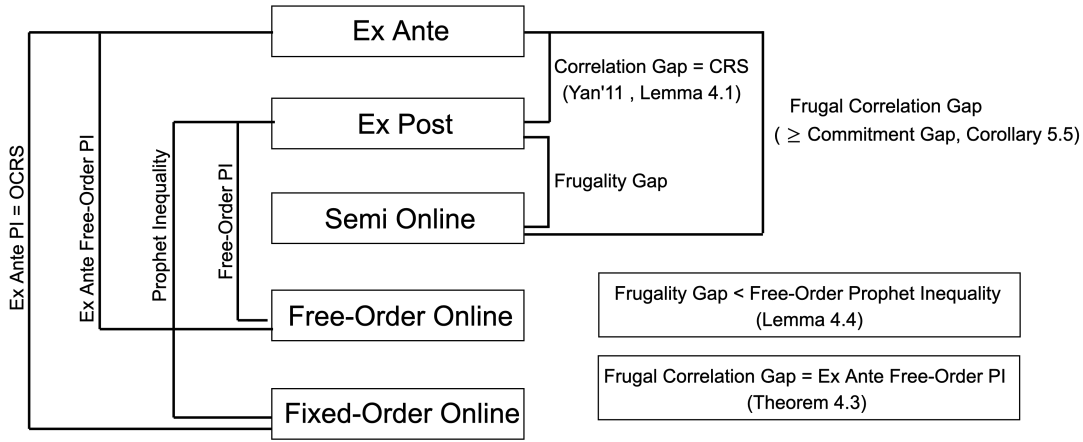
\begin{figure}[!ht]
\centering
\resizebox{\textwidth}{!}{%
\begin{circuitikz}
\draw  (6,13) rectangle (6,13);
\draw  (3.5,12.25) rectangle  node {Ex Ante} (7.75,11.25);
\draw  (3.5,10.5) rectangle  node {Ex Post} (7.75,9.5);
\draw  (3.5,8.75) rectangle  node {One-Sided Semi-Online} (7.75,7.75);
\draw  (3.5,7) rectangle  node {Free-Order Online} (7.75,6);
\draw  (3.5,5.25) rectangle  node {Fixed-Order Online} (7.75,4.25);
\draw (3.5,11.75) to[short] (0.5,11.75);
\draw (0.5,11.75) to[short] (0.5,4.5);
\draw (0.5,4.5) to[short] (3.5,4.5);
\draw (1.25,11.75) to[short] (1.25,6.25);
\draw (1.25,6.25) to[short] (3.5,6.25);
\draw (3.5,10) to[short] (2,10);
\draw (2,10) to[short] (2,5);
\draw (2,5) to[short] (3.5,5);
\draw (2.75,10) to[short] (2.75,6.75);
\draw (2.75,6.75) to[short] (3.5,6.75);
\node [font=\normalsize, rotate around={90:(0,0)}] at (2.5,8.25) {Free-Order PI};
\node [font=\normalsize, rotate around={90:(0,0)}] at (1.75,8.25) {Prophet Inequality};
\node [font=\normalsize, rotate around={90:(0,0)}] at (1,8.25) {Ex Ante Free-Order PI};
\node [font=\normalsize, rotate around={90:(0,0)}] at (0.25,8.25) {Ex Ante PI = OCRS};
\draw (7.75,11.75) to[short] (14,11.75);
\draw (14,11.75) to[short] (14,8.25);
\draw (14,8.25) to[short] (7.75,8.25);
\draw (8.5,11.75) to[short] (8.5,10.25);
\draw (8.5,10.25) to[short] (7.75,10.25);
\draw (7.75,9.75) to[short] (8.5,9.75);
\draw (8.5,9.75) to[short] (8.5,8.25);
\node [font=\normalsize,align=left] at (11,11) {Correlation Gap \\ = CRS \\ \quad(\cite{Y11}, \Cref{lemma:expost-vs-exante})};
\node [font=\normalsize,align=left] at (10.1,9) {Frugality \\$>$ Free-Order PI \\
\quad(\Cref{lemma:frugal-vs-semionline})};
\node [font=\normalsize,align=left] at (16.35,10) {Frugal Correlation Gap \\ $=$ Ex Ante Free-Order PI \\ \quad (\Cref{lemma:semi-online-and-eafopi}) \\
$\le$ Commitment Gap \\ \quad (\Cref{cor:eafopi-bounds-comgap})};
\end{circuitikz}
}%
\caption{Connections between different benchmarks for Bayesian Combinatorial Selection.}
\label{fig:gaps}
\end{figure}


    
\begin{table}[ht!]
\centering
\resizebox{\textwidth}{!}{%
\renewcommand{\arraystretch}{1.2}
\begin{tabular}{|c|c|c|c|c|}
\hline
Feasibility     & Frugality & Correlation gap & Commitment gap & Ex ante prophet inequality\\
constraint & ~ & ($=$ CRS) & ($\ge$ Frugal correlation gap) & ($=$ OCRS) \\ 
~ & ~ & ~ & ($=$ ex ante free-order PI) & ~ \\ \hline
Single Selection & $1$           & $1-\frac{1}{e}$       & $1-\frac{1}{e}$    & $\frac{1}{2}$               \\ \hline
Matroid & $1$           & $1-\frac{1}{e}$       & $1-\frac{1}{e}$          & $\frac{1}{2}$         \\ \hline
Knapsack                  & $\frac{1}{2}$           & $0.319$             & $0.319$    & $0.319$                \\ \hline
$k$-system               & $\frac{1}{k}$           & $\frac{1}{k+1}$           & $\frac{1}{k+1}$     & $\frac{1}{e(k+1)}$             \\ \hline
Bipartite Matching               & $0.659$           & $0.509$         &  $0.478$      & $0.349$             \\ \hline
General Matching               & $0.535$           & $0.474$           & $0.474$    & $0.344$             \\ \hline
\end{tabular}
\renewcommand{\arraystretch}{1}
}%
\caption{Known results on different gaps for various families of constraints.}
\label{table:res}
\end{table}

\subsubsection*{CICS versus Prophet Inequalities}

It is worth noting the differences between prophet inequalities and cost of information models like the CICS. On the one hand, in prophet inequalities elements can be queried for free, whereas in CICS we need to pay to query an element. On the other hand, in a prophet inequality the algorithm must make select/reject decisions for each element right after it is queried, whereas in CICS the algorithm is allowed to go back to a previously queried element to explore it further or select it. \cite{S17}'s work argued that the first difference is not too important; MC-CICS can be reduced to a ``free information'' setting essentially without loss. Our work argues that the second difference (online versus one-sided semi-online selection) is not too important when we compare against the ex ante relaxation; in the worst case the two are exactly the same.

\subsubsection*{Connections to the Query-Commit Model}

One-sided semi-online algorithms are closely related to the {\em query-commit} model that has been studied extensively in the context of bipartite and general matchings. We make this connection explicit and use it to obtain improved bounds for these feasibility constraints. 
The query-commit model was 
first introduced by~\cite{chen2009approximating}. Here, a stochastic matching algorithm is given a graph $G = (V, E)$, where every edge $e \in E$ has a known weight $w_e \ge 0$ and exists independently with known probability $p_e$. During execution, the algorithm may (adaptively) query any edge to determine whether it exists. If the edge exists, the algorithm must commit to including this edge in its output set. At the end, the algorithm must output a valid matching of the graph $G$.

A query-commit algorithm bears many similarities with our one-sided semi-online algorithm. In the context of the Bayesian Combinatorial Selection problem, a one-sided semi-online algorithm on an instance where every distribution is Bernoulli behaves exactly as that of a query-commit algorithm, and therefore query-commit algorithms inherit the best competitive ratio for one-sided semi-online algorithms. In~\Cref{sec:query-commit-semi-online}, we essentially showed the converse statement: any $\alpha$-competitive query-commit algorithm can be used in a black-box manner to design an $(\alpha-\eps)$-competitive one-sided semi-online algorithm. Such a reduction has been shown for specific query-commit algorithms before~\citep{gamlath2019beating,derakhshan2023beating,fu2021random}; however, to the best of our knowledge, we are the first to generalize such a result via a black-box reduction to any query-commit algorithm. In turn, this gives us improved results for the frugality of bipartite matching and general matching constraints.

\subsubsection*{Computational Considerations}

We conclude by summarizing how our approach yields efficient approximation algorithms for all CICS instances whose commitment gaps are bounded in~\Cref{table:res}. Given a CICS instance, the first step is access to the ex ante optimal commitment; in~\Cref{sec:cics-exante}, we show that this solution can be computed exactly and efficiently, yielding a committing policy. We then compute the surrogate value for each Markov chain that we obtain under this commitments, using efficient algorithms from~\citep{DTW03, CCHS24}. Finally, we feed these surrogate values into a one-sided semi-online algorithm, as described in~\Cref{sec:bcs-to-cics}, to obtain the final policy. Thus, whenever the one-sided semi-online algorithm itself is efficiently implementable---which holds for all constraints listed in~\Cref{table:res}---our bounds on the commitment gap translate directly into efficiently computable committing policies.

\subsection{References for~\Cref{table:res}}
For single selection, matroid, knapsack, and $k$-system constraints, their frugalities are due to~\cite{S17}, who introduced the class of frugal algorithms for deterministic selection problems. The correlation gap of $1 - \frac{1}{e}$ for matroid constraints was first established by~\cite{CCPV07,ADSY12}; this immediately implies the same bound for the frugal correlation gap and ex ante free-order prophet inequalities since the frugality is $1$. The first ex ante free-order prophet inequality for $k$-systems was given by~\cite{CHMS10} for the special case of intersections of $k$-matroids, and later generalized by~\cite{Y11} for arbitrary $k$-systems. Both a $\frac{1}{2}$-OCRS (or equivalently ex ante prophet inequality) for matroid constraints and a $\frac{1}{e(k+1)}$-OCRS for $k$-systems can be found in~\cite{FSZ15}. All the other bounds for knapsack are due to~\cite{JMZ21} who define an OCRS for a more challenging version where the sizes of the elements are also stochastic; this immediately implies the same bound both on the frugality and the frugal correlation gap.

For the case of the matching constraints,~\cite{nuti2025towards} gives a $0.509$-selectable CRS for bipartite matching, implying a bound on the correlation gap.~\cite{macrury2025random} show that there exists a $0.478$-selectable and $0.474$-selectable random-order contention resolution scheme (RCRS) -- and therefore an ex ante free-order prophet inequality -- for bipartite matching and general matching respectively, implying a lower bound for the commitment gaps of these constraints. The $0.474$-selectable RCRS also provides the best known bound on the correlation gap for general matching. In the same paper,~\cite{macrury2025random} also show that the OCRS of~\cite{ezra2022prophet} is $0.349$-selectable and $0.344$-selectable for bipartite matching and general matching, respectively. Finally, the bounds on the frugality are obtained via our reduction from the one-sided semi-online model to the query-commit model. For bipartite matchings, a $0.659$-competitive query-commit algorithm has been discovered~\citep{huang2025edge}. Meanwhile, in the general graph case,~\cite{macrury2024random} provide a $0.535$-selectable \emph{vertex-arrival} 
RCRS for general matching. Combined with a reduction from query-commit to vertex-arrival RCRS~\citep{fu2021random}, this implies a $0.535$-competitive query-commit algorithm for general matching.


\subsection{Related Work} 

\subsubsection*{Combinatorial Selection with Costly Information}

Our approach is closely related and inspired by the recent contributions of~\citep{BLW25}, where an efficient and ``universal'' $\frac{1}{2}$-approximation for CICS under matroid constraints is developed. This work introduces the idea of efficiently solving for the ex ante optimal policy of the CICS instance and then ``rounding'' it into an ex post feasible solution. The rounding is performed by exploiting an ex ante prophet inequality for the constraint, transferring the approximation factor. Importantly, the produced policy is not necessarily a committing one, and a factor of $\frac{1}{2}$ is necessarily lost in approximation by requiring an ex ante prophet inequality (even for single selection). In contrast, our approach provides a direct bound on the commitment gap and exploits an algorithm with a small frugal correlation gap, allowing us to obtain an improved $1-\frac{1}{e}$ approximation for CICS under matroid constraints. Furthermore, we improve the computational results of~\cite{BLW25} by showing that we can efficiently compute the ex ante optimal policy \textit{exactly} (compared to the FPTAS they design) and for \textit{any} MDP (compared to the special case of DAG-MDPs they consider).

Prior to the work of~\cite{BLW25}, most known results for CICS applied to specific settings that are mostly generalizations of the Pandora's box problem, see for example \citep{GMS08, KWG16, S17, D18, BK19, EHLM19, GJSS19, BFLL20, AJS20, FLL22, BC22, BEFF23, BW24, DS24,CCHS24}. For more details on these variants, we defer the reader to the recent survey of~\citet{BC24-survey}.

The study of the commitment gap for CICS was initiated by~\cite{CCHS24}. By extending the definition of surrogate values to MDPs, they are able to obtain ``local'' approximation conditions that characterize the quality of each commitment and compose into upper bounds for the commitment gap. Importantly, their conditions are instance-dependent and rely on both the structure of the underlying MPDs and the feasibility constraint; in contrast, our approach provides more universal bounds that only depend on the combinatorial constraint. By instantiating their framework with different variants of CICS, they are able to obtain novel bounds on the commitment gap. While our bounds for the maximization setting outperform their applications, they also obtain several bounds on the commitment gap for the minimization version of CICS, whereas our framework cannot produce any meaningful results in the minimization setting (\Cref{app:minimization}).

Finally, of particular interest is the most well-studied variant of max-CICS, namely the Pandora's Box problem with Optional Inspection. In this setting, the decision-maker is allowed to select a closed box, without paying to open it. Under single-selection constraints, the problem is NP-Hard but a PTAS is known to exist \citep{FLL22,BC22}. Under (general) matroid constraints, \cite{BK19} show that the commitment gap is upper bounded by $1-\frac{1}{e} \approx 0.632$ via a reduction to submodular maximization. Their approach is based on a reduction to nonadaptive stochastic submodular maximization, maximizing the weighted rank function of the matroid constraint subject to another partition matroid. While achieving exactly a $1 - 1/e$ ratio in the submodular maximization step is nonconstructive, it can be made polynomial-time constructive at the cost of arbitrarily small error \citep{AN16}. Our results show that the lower bound of \cite{BK19} can indeed be matched exactly and efficiently by employing an ex ante relaxation of the problem and rounding via a free-order prophet inequality.

\subsubsection*{Bandit Superprocesses}

CICS is closely related to a class of sequential decision-making problems known as bandit superprocesses (BSPs), originally formulated by \citet{N73}. In a BSP, the decision-maker controls multiple independent MDPs, selecting one to advance at each time step. Upon doing so, they receive a discounted reward, and the selected MDP transitions stochastically. A key distinction between the two frameworks is in their objectives and structure: while a BSP may operate over an infinite horizon, a CICS algorithm must eventually terminate by selecting a feasible subset of alternatives.

When the constituent MDPs in a BSP are simple Markov chains, the problem reduces to the classical and extensively studied Multi-Armed Bandit (MAB) model. A seminal result by \citet{G79} shows that MABs admit a remarkably elegant optimal policy: assign an index to each state of every arm independently of the others, and at each step, choose to advance the arm currently holding the highest index. This powerful principle of indexability has been shown to extend to finite-horizon settings by \citet{DTW03}, and to more complex combinatorial variants by \citet{GJSS19}. On the other hand, BSPs defined over general MDPs are not indexable \citep{G82}. Nevertheless, the local structure of the problem provides insight into the optimal solution: given an MDP $\mc$, \cite{W80}
considers a local problem $(\mc, y)$, where at every step the algorithm can either advance $\mc$ or terminate with at a reward of $y$. The values of these local games, formalized as ``optimality curves'', one for each constituent MDP in the BSP, can be combined to obtain the Whittle integral, a lower bound on the global optimum \citep{W80,BS13}. \cite{CCHS24} provide a new interpretation for Whittle's optimality curves by connecting them with a mapping from trajectories in the MDP to surrogate values.

\subsubsection*{Prophet Inequalities, Correlation Gap and Contention Resolution}

Prophet inequalities originate from classical optimal stopping theory. The foundational results, developed by~\cite{krengel1977semiamarts} and~\cite{samuel1984comparison}, show that in the basic setting, it is possible to achieve at least half of the expected reward obtained by the prophet. The last decade has seen a renewed surge of interest in combinatorial prophet inequalities, which generalize the classical model to more complex feasibility constraints. A seminal result by~\cite{kleinberg2012matroid} introduced a constant-competitive algorithm for selecting a maximum-weight independent set in a matroid, using carefully designed threshold policies. 

Prophet inequalities are also deeply connected to mechanism design, particularly in Bayesian settings. \cite{CHMS10} demonstrated that sequential posted-price mechanisms can approximate optimal revenue in multi-item Bayesian auctions, using prophet-like arguments; this result was later strengthened by~\cite{Y11}. These insights have led to a growing body of work that uses prophet inequalities as a tool for designing truthful and approximately optimal mechanisms under uncertainty.
For more details in recent developments and applications, we defer the reader to the surveys of~\citep{L17,CFHOV19}.


A semi-online model of selection for BCS was previously introduced by~\cite{HS23}. In \citeauthor{HS23}'s \emph{semi-online prophet inequality}, the algorithm proposes a single threshold for every element, queries {\em all} elements to determine whether their value is larger than the corresponding threshold, and at the end selects a feasible subset of elements that maximizes the total conditional expected value. In particular, this class of algorithms differs from ours in two ways: 1) a single threshold is proposed for each variable, contrasting with the possibly multiple thresholds that a one-sided semi-online algorithm may query per element, and 2) selection is carried out after all the queries, unlike ours where an element must be selected immediately if it is above the queried threshold. We use the term {\em one-sided} for our model to emphasize these differences.


The correlation gap, introduced by~\cite{ADSY12}, measures the worst-case ratio between the expectation of a function under independent sampling versus an arbitrary correlated distribution with the same marginals. They showed that for monotone submodular functions, the correlation gap is bounded by $1 - \frac{1}{e}$ and for certain classes of linear packing problems, it is constant. These bounds were later used to obtain improved approximations for submodular maximization. The correlation gap has since become an essential analytical tool in areas involving randomized rounding, submodular maximization, and prophet inequalities.

Contention resolution schemes were introduced by~\cite{CVZ11} as a rounding tool for maximizing submodular functions subject to combinatorial constraints. Their framework provided a general approach to converting fractional solutions—typically obtained via the multilinear relaxation—into integral ones while preserving approximation guarantees. They also show that the optimal contention resolution scheme achieved for a constraint equals its correlation gap. \cite{FSZ15} later apply this definition to online rounding processes by defining online contention resolution schemes (OCRSes); \cite{LS18} show that designing an OCRS is equivalent to obtaining an ex ante free-order prophet inequality and vice-versa.

\subsection{Organization}
We formally define Combinatorial Selection with Costly Information (CICS) and related concepts in~\Cref{sec:cics}. In~\Cref{sec:bcs}, we introduce Bayesian Combinatorial Selection (BCS); define the different benchmarks we consider in this work and the gaps between them; and introduce the family of one-sided semi-online selection algorithms. In~\Cref{sec:bcs-to-cics}, we formally connect CICS to BCS and provide a lower bound for the commitment gap via the frugal correlation gap. In~\Cref{sec:cics-exante}, we discuss how to compute the ex ante optimal policy for CICS. Finally, in~\Cref{sec:query-commit-semi-online}, we discuss the additional connection between our one-sided semi-online algorithms and the query-commit model in the context of matching constraints. Omitted proofs are presented in~\Cref{app:proofs}. 
\section{Notation}

Throughout this paper, we use $[n]:=\{1,2,\cdots , n\}$ to denote a ground set of elements. The goal of the algorithm in each of the settings we consider is to pick a feasible subset of the ground set. We use $\feas\subseteq 2^{[n]}$ to denote the feasibility constraint and assume that it is downwards closed, that is, $A\subset B$ and $B\in\feas$ implies $A\in\feas$. We use $\polytope(\feas)$ to denote the polytope relaxation of $\feas$, i.e. the convex hull of the vectors $\{\mathbbm{1}_S : S\in\feas\}$ where $\mathbbm{1}_S\in\R^n$ denotes the indicator of a set $S\subseteq [n]$. 

The value $X_i$ of each element $i\in [n]$ is stochastic and we use $\disti_i$ to denote its distribution, with $\dist = \disti_1\times\cdots \times \disti_n$ denoting the corresponding product distribution. We use $x_i$ to denote the instantiation of $X_i$. In all of the settings we study, the algorithm chooses a subset $S(x)\subseteq [n]$ for every instantiation $x\sim\dist$ and obtains a total expected value of $\operatorname{E}_{x\sim\dist}[\sum_{i\in S(x)} x_i]$. Different settings, benchmarks, and algorithms differ in the constraints satisfied by the selected set $S(x)$ and the amount of information available for this decision making. Finally, for any distribution $\disti_i$ over the reals and any $q_i\in [0,1]$, we use $F_{\disti_i}(q_i)$ to denote the expected value of a realization $x_i\sim\disti_i$, conditioned on it lying in the top $q_i$ quantiles of $\disti_i$.\footnote{If $\disti_i$ contains a point mass at the $q_i$-th quantile, we assume that this point mass is included in the conditional expectation with the right probability, so that the probability $q_i$ is met exactly.}



\section{Combinatorial Selection with Costly Information}\label{sec:cics}

Combinatorial Selection with Costly Information is defined over $n$ alternatives and a downwards-closed feasibility constraint $\feas\subseteq 2^{[n]}$. The costly information component reflects the fact that the value of each alternative $i\in [n]$ is not known a priori and must be learned through a sequence of costly actions. This learning process is formalized as a Markov decision process (MDP) $\mc_i$. 

An MDP $\mc$ is defined over a set of states $\Omega$; each state $s\in \Omega$ admits a set of actions $a\in A(s)$ and we use $c(a)$ to denote their costs. Taking an action $a\in A(s)$ from a state $s\in \Omega$ results in a random transition to a new state $s'\in \Omega$ according to some known distribution $\Pi(s,a)$. We assume that all MDPs contain a (non-empty) set of terminal states $T\subseteq \Omega$ that are sink states, i.e. no more actions can be taken as soon as these states are reached; each terminal state $t\in T$ has an associated value $v(t)$. Finally, we use $\inst = (\feas,\mset)$ to denote the CICS instance where $\mset = (\mc_1,\cdots,\mc_n)$ is the collection of the underlying MDPs, one for each alternative $i\in [n]$. 

Let $\inst = (\feas,\mset)$ be any CICS instance. A policy $\pol$ for $\inst$ is a sequential algorithm that will eventually accept a subset of the MDPs and halt, collecting a reward, while suffering a loss for all the costly actions that it will take during this process. At any point in time, each MDP occupies one of its states, and only MDPs that occupy terminal states can be accepted; the algorithm is allowed to advance the state of an MDP through an available action by paying the corresponding cost. Formally, the state of $\pol$ is initialized at $s:=(\sigma_1,\cdots , \sigma_n)$ where each $\sigma_i\in \Omega_i$ is a pre-specified ``root'' state for the $i$-th MDP. Then, at every step, policy $\pol$ proceeds to take one of two options:
\begin{enumerate}
    \item Pick an index $i\in [n]$ such that $s_i\in\Omega_i\setminus T_i$ is a non-terminal state and an action $a_i\in A_i(s_i)$. Pay the action cost $c_i(a_i)$, observe a new state $s'_i$ that is generated from the distribution $\Pi_i(s_i,a_i)$ and iterate over the new set of states $s=(s_{-i},s'_i)$.

    \item Select a feasible subset $S\subseteq \feas\cap \{i:s_i\in T_i\}$ of the MDPs whose current state is terminal. Then, collect reward $\sum_{i\in S}v_i(s_i)$ and halt.
\end{enumerate}
Let $C_\pol$ and $V_\pol$ respectively denote the random variables corresponding to the total cost paid by policy $\pol$ while taking actions (option 1) and the reward it collected (option 2) before halting. We evaluate the performance of policy $\pol$ with respect to its utility $\util_{\inst}(\pol):=\expectt{V_\pol-C_\pol}$ and we use $\opt(\inst):=\max_\pol[\util_{\inst}(\pol)]$ to denote the utility of the optimal policy for the CICS instance $\inst$. Finally, we use $S_\pol$ to denote the (random) set of MDPs that a policy $\pol$ accepts; note that $\probb{S_\pol\in\feas}=1$ for any feasible policy $\pol$ by definition.

\subsection*{Committing Policies} 
Committing policies are a family of algorithms for CICS that make all local decisions within the MDPs in advance. In particular, a committing policy is described with respect to a commitment $\comms=(\com_1,\cdots,\com_n)$ where each $\com_i$ is a mapping from the non-terminal states $s_i\in \Omega_i\setminus T_i$ to distributions over available actions $\com_i(s_i)\in\Delta(A_i(s_i))$. We say that a policy commits to $\comms$ if whenever it takes option 1 in the above description, it only select the index $i\in [n]$ and then the action $a_i$ is automatically sampled from $\com_i(s_i)$. We use $\comset(\comms)$ to denote the set of all policies that commit to $\comms$ and $\comset(\inst)$ to denote the set of all possible commitments $\comms$ for a CICS instance $\inst$. Then, the commitment gap is defined as the ratio between the utilities of the overall optimal policy for the instance and the optimal committing policy.
\begin{definition}[Commitment Gap]
    Let $\inst =(\feas,\mset)$ be any CICS instance. The commitment gap of $\inst$ is defined as
    \[\comgap(\inst):= \frac{\max_{\comms\in\comset(\inst)}\max_{\pol\in\comset(\comms)}\util_{\inst}(\pol)}{\max_\pol\util_{\inst}(\pol)}.\]
\end{definition}

Finally, we note that all policies that commit to some $\comms\in\comset(\inst)$ can be equivalently viewed as policies for a ``reduced'' CICS instance over Markov chains. A Markov chain corresponds to the special case of an MDP for which $|A(s)|=1$ for all $s\in \Omega\setminus T$ i.e. all non-terminal states have a unique available action. We use MC-CICS to denote CICS instances that are defined over Markov chains. Given any MDP $\mc$ and a mapping $\com(\cdot)$ from states to distributions over actions, we can define a Markov chain $\mc^\com$ over the same terminal and non-terminal states as $\mc$, such that the unique action from a state $s\in \Omega\setminus T$ has cost $\sum_{a\in A(s)}\com(s,a)c(a)$ and transitions to each state $s'$ with probability $\sum_{a\in A(s)}\com(s,a)\Pi(s,a,s')$. Then, any policy for a CICS instance $\inst=(\feas,\mset)$ that commits to $\comms=(\com_1,\cdots,\com_n)$ can be viewed as a policy for the MC-CICS instance $\inst^\comms = (\feas,\mset^\comms)$ where $\mset^\comms=(\mc_1^{\com_1},\cdots,\mc_n^{\com_n})$ and vice-versa.

\subsection*{Ex Ante Feasible Policies}

Up next, we consider an ex ante relaxation for CICS, first defined by~\cite{BLW25}. Let $\inst=(\feas,\mset)$ be any CICS instance. Recall that any policy $\pol$ for $\inst$ needs to satisfy (ex post) feasibility: the set of MDPs it accepts needs to be feasible ($S_\pol\in\feas$) with probability $1$. We now relax this requirement by considering ex ante feasible policies where the feasibility constraint must be satisfied only in expectation, i.e. we require $\mathbb{E}[\mathbbm{1}_{S_\pol}]\in\polytope(\feas)$. We use $\exanteopt(\inst)$ to denote the utility of the optimal ex ante feasible policy for $\inst$; since any ex post feasible policy is also ex ante feasible, we clearly have that $\opt(\inst)\leq \exanteopt(\inst)$.

It turns out that computing the optimal ex ante feasible policy for a CICS instance $\inst$ is a much easier task than computing the optimal ex post feasible policy. \cite{BLW25} prove that in the special case of MDPs that do not contain any directed cycles, the optimal ex ante feasible policy is a committing one, and there exists a FPTAS for finding ex ante optimal policies. By using results from the constrained MDP literature, we are able to strengthen this result and achieve an exact efficient computation of the optimal ex ante feasible policies in general MDPs under the same mild assumption on the constraint $\feas$ that \cite{BLW25} makes; namely, that it admits an efficient separation oracle.

\begin{lemma}\label{lemma:exante-compute}
    Let $\inst = (\feas,\mset)$ be any CICS instance and assume that $\feas$ admits an efficient separation oracle. Then, we can efficiently compute a commitment $\comms\in\comset(\inst)$ such that \[\exanteopt(\inst)=\exanteopt(\inst^\comms).\]
\end{lemma}

We prove~\Cref{lemma:exante-compute} in~\Cref{sec:cics-exante}. While ex ante feasible policies are not valid solutions for CICS, the above lemma produces a good commitment for $\inst$ and essentially reduces bounding the commitment gap to relating the utilities of ex post and ex ante feasible policies in MC-CICS.

\subsection*{Surrogate Values} 
Finally, we require some statements and technical results from~\cite{GJSS19} and~\cite{CCHS24} that apply to MC-CICS instances. For ease of presentation, we will only provide the necessary definitions and theorems here; we defer the reader to Section 3 of~\cite{GJSS19} and Section 3.1 of~\cite{CCHS24} for the full details and intuition behind these claims. In the following, $\trajset(\mc)$ denotes the set of all possible trajectories starting from the root state of a Markov chain $\mc$ and ending in a terminal state, $p(\traj)$ denotes the probability that a trajectory $\traj\in\trajset(\mc)$ gets realized and $t(\traj)$ denotes the terminal state at the end of trajectory $\traj$.

We begin with the notion of surrogate values for a Markov chain $\mc$, which essentially amortize the action costs $c(a)$ to trajectories $\traj\in \trajset(\mc)$. In particular, when a policy accepts an alternative $i$ resulting from a trajectory $\traj_i\in\trajset(\mc_i)$, the difference of the actual value of the alternative and the surrogate value of the trajectory, namely $v(t(\traj_i))-w(\traj_i)$, goes towards the cost of the actions that the policy took in $\mc_i$ to arrive at the terminal. \cite{CCHS24} define this amortization through a water-draining procedure so that a policy never overpays action costs in the special case of Markov chains that do not contain any directed cycles. \cite{GJSS19} provide a more general definition that applies to arbitrary Markov chains. Furthermore, if the policy satisfies an additional ``promise of payment'' property defined below, then the surrogate values fully account for the action costs and provide an accurate expression for the policy's final expected utility.

\begin{definition}[Surrogate Values]
    Any Markov chain $\mc$ defines a mapping $g:\Omega\mapsto\R$ from its states $s\in \Omega$ to grades $g(s)\in\R$. Furthermore, the index of a trajectory $\traj\in\trajset(\mc)$ is defined as $w(\gamma):=\min_{s\in\traj}g(s)$. Finally, the surrogate value $W$ of the Markov chain $\mc$ is the random variable that takes value $w(\traj)$ with probability $p(\traj)$.
\end{definition}

\begin{definition}[Promise of Payment]
     We say that a policy $\pol$ for an MC-CICS instance $\inst$ satisfies promise of payment if in every possible trajectory of the policy and for every alternative $i$, when the policy halts one of the following properties holds: (i) $i$ is selected (i.e. $i\in S_\pol$), or, (ii) the Markov chain $\mc_i$ is never advanced, or, (iii) the final state in $\mc_i$ has the minimum grade across all states visited in $\mc_i$.
\end{definition}


\noindent
The following relates the utility of a policy that satisfies promise of payment to its surrogate value.

\begin{lemma}\label{lemma:amortization-bound} (\cite{GJSS19})
     The utility of any policy $\pol$ for a MC-CICS instance $\inst=(\feas,\mset)$ satisfies 
     \[\util_{\inst}(\pol)\leq \sum_{i=1}^n\sum_{\traj_i\in \trajset(\mc_i)} \probb{\traj_i \text{ realized in }\mc_i \text{ and }\pol\text{ accepts }t(\traj_i)}\cdot w_i(\traj_i).\]
     Furthermore, the above inequality holds with equality if $\pol$ satisfies promise of payment.
\end{lemma}

We note that variants of the above lemma have been stated in many previous works, including~\cite{DTW03,GJSS19,CCHS24} and~\cite{BLW25}. While the term Promise of Payment was introduced by~\cite{CCHS24} and their results are stated only for acyclic Markov chains, the same condition is proven by~\cite{GJSS19} for general Markov chains using techniques from~\cite{DTW03} under a different notation. We replicate the proof in the appendix for completeness.

\section{Bayesian Combinatorial Selection}\label{sec:bcs}

In this section, we formally define the stochastic optimization problem of Bayesian Combinatorial Selection (henceforth, BCS), the benchmarks and families of algorithms we consider over it, and the relationships between these benchmarks. 


\begin{definition}[BCS and Benchmarks]
    An instance $(\feas,\dist)$ of BCS consists of a constraint $\feas\subseteq 2^{[n]}$ and a product distribution $\dist = \disti_1\times\cdots \times \disti_n$ over $n$ stochastic elements. Given such an instance, we consider the ex post optimal and ex ante optimal benchmarks, respectively defined as \[\expost(\feas,\dist):=\expect{x\sim\dist}{\max_{S\in\feas}\sum_{i\in S}x_i}\quad\text{,}\quad \exante(\feas,\dist) = \max_{q\in\polytope(\feas)}\sum_{i=1}^n q_iF_{\disti_i}(q_i).\]
\end{definition}

Note that $\expost(\feas,\dist)$ captures the expected value of the hindsight optimal algorithm that for every $x\sim\dist$ selects the value-maximizing set $S(x)\in\feas$. The ex ante optimum captures the expected value of an algorithm that selects a subset $S(x)\subseteq [n]$ for every $x\sim\dist$, but this subset is not required to be feasible at all vectors $x$. Instead, we require that $S(x)$ satisfies feasibility in expectation, namely that $\expect{x\sim\dist}{\mathbbm{1}_{S(x)}}\in\polytope(\feas)$. Denoting $\prob{x\sim\dist}{i\in S(x)}$ by $q_i$ provides the above expression for 
$\exante(\feas,\dist)$.
Clearly, ex post feasibility implies ex ante feasibility and thus $\expost(\feas,\dist)\leq \exante(\feas,\dist)$ for all constraints $\feas$ and distributions $\dist$. 

Given a constraint $\feas$, it will be useful to consider the minimum ratio between the ex post and the ex ante optima, over all possible distributions $\dist$. It turns out that this quantity is precisely the same as the correlation gap of $\feas$, as introduced by~\cite{ADSY12}\footnote{\cite{ADSY12} define the correlation gap with respect to a set function; \cite{CVZ11} later extend this definition to a constraint $\feas$ by considering the correlation gap of the associated weighted rank function.}.

\begin{definition}[Correlation Gap] For a feasibility constraint $\feas\subseteq 2^{[n]}$ and $x\in\polytope(\feas)$, let $R(x)\subseteq [n]$ denote the random set obtained by picking each element $i\in [n]$ independently with probability $x_i$.
    The correlation gap of $\feas$ is defined as
    \[\corgap(\feas):= \min_{x\in\polytope(\feas),y\in\R^n_+}\frac{\expectt{\max_{S\in\feas\cap R(x)}\sum_{i\in S}y_i}}{\sum_{i=1}^nx_iy_i}.\]
\end{definition}

The following lemma relating the correlation gap to the ex ante and ex post optima for BCS was first proved by \cite{Y11}.\footnote{\cite{CVZ11} later showed that the correlation gap exactly captures the performance of offline contention resolution schemes (CRS).} We reproduce the proof in the appendix for completeness.

\begin{lemma}\label{lemma:expost-vs-exante}
For any constraint $\feas$, we have $\corgap(\feas)=\min_\dist\big(\expost(\feas,\dist)/\exante(\feas,\dist)\big)$.
\end{lemma}


\paragraph{Online Selection Algorithms.} We will now consider ex post feasible selection algorithms for BCS that observe the values of the elements sequentially and must commit to accepting or rejecting elements before observing all of the values. We distinguish between two kinds of algorithms. 

The first probes elements in arbitrary order, learning their values exactly, but must decide to accept or reject each element irrevocably upon observing its value. We call this a {\em Free-Order Online Selection} algorithm. The second can visit each element multiple times in arbitrary order and at each visit determines an acceptance threshold. At each visit the algorithm only learns whether the value of the element exceeds the threshold or not, and if it does, the algorithm must accept the element before moving forward. We call this a {\em One-Sided Semi-Online Selection} algorithm.


\begin{definition}[Free-Order Online Algorithms]
    A free-order online selection algorithm for a BCS instance $(\feas,\dist)$ proceeds as follows. Starting from $S=R=\emptyset$, it iteratively selects an index $i\in [n]\setminus (S\cup R)$ with $S\cup\{i\}\in\feas$ and a threshold $\tau_i$. It observes the instantiated value $x_i$. If $x_i\geq \tau_i$ it accepts the element $(S:=S\cup\{i\})$, otherwise it rejects it $(R:=R\cup\{i\})$. The process continues until no more elements can be selected or the algorithm decides to halt.
\end{definition}    

\begin{definition}[One-Sided Semi-Online Algorithms]
    A one-sided semi-online selection algorithm for a BCS instance $(\feas,\dist)$ proceeds as follows. Starting from $S=\emptyset$, it iteratively selects an index $i\in [n]\setminus S$ with $S\cup\{i\}\in\feas$ and a threshold $\tau_i$. It observes whether $x_i\ge \tau_i$ or not. If $x_i\geq \tau_i$ it accepts the element, that is, $S:=S\cup\{i\}$. The process continues until no more indices can be selected or the algorithm decides to halt.
\end{definition}
    
    The expected total value obtained by an algorithm $\so$ for the BCS instance $(\feas,\dist)$ is denoted by $\val_{\feas,\dist}(\so)$. We use $\freeorder(\feas,\dist)$ and $\semi(\feas,\dist)$ to denote the value of the optimal free-order online and one-sided semi-online algorithms, respectively, for the BCS instance $(\feas,\dist)$. 
    
\paragraph{Online Selection  Algorithms and Prophet Inequalities.}
    Observe that Free-Order Online algorithms are precisely the class of algorithms whose performance is captured by a Free-Order Prophet Inequality.\footnote{A (regular) Prophet Inequality, in contrast, specifies a fixed (adversarial) ordering over the elements and requires the algorithm to probe elements in that order and make select/reject decisions online.} In particular, we will say that a constraint $\feas$ admits an $\alpha$ free-order (resp. ex ante free-order) prophet inequality for some $\alpha\leq 1$ if and only if
    \[\min_\dist\frac{\freeorder(\feas,\dist)}{\expost(\feas,\dist)}\geq \alpha \quad \quad \bigg(\text{resp. } \min_\dist\frac{\freeorder(\feas,\dist)}{\exante(\feas,\dist)}\geq \alpha \bigg).\]
 
Clearly, any free-order online algorithm can be simulated by a one-sided semi-online algorithm. Therefore, we have that $\semi(\feas,\dist)\ge\freeorder(\feas,\dist)$ for all BCS instances $(\feas, \dist)$. In fact, when compared against the ex post optimal benchmark, there is a strict separation between the two classes of algorithms. Consider, for example, the objective of selecting a single largest value element. A one-sided semi-online algorithm can match the ex post optimum for this objective by iteratively probing the element whose support has the largest maximum value and using this value as the threshold. On the other hand, it is well-known that we cannot achieve an $\alpha$ free-order prophet inequality with $\alpha > 0.745$ even for single-element selection when all the distributions are identical~\citep{HK82}. We capture this observation in the following:

\begin{fact}
    Let $\feas=\{S\subseteq [n]: |S|\leq 1\}$. Then, for any product distribution $\dist$, we have that $\semi(\feas,\dist)=\expost(\feas,\dist)$, but there exists a $\dist$ with $\freeorder(\feas,\dist)\le 0.745\,\expost(\feas,\dist)$.
\end{fact}

Surprisingly, in contrast, when comparing against the {\em ex ante} optimal benchmark, this separation collapses and free-order online algorithms match the performance of one-sided semi-online algorithms.
\begin{theorem}\label{lemma:semi-online-and-eafopi}
For any constraint $\feas$, we have 
    \begin{align*}
    \text{Frugal Correlation Gap}(\feas) := & \min_\dist \left(\frac{\semi(\feas,\dist)}{\exante(\feas,\dist)}\right) \\
    & = \min_\dist \left(\frac{\freeorder(\feas,\dist)}{\exante(\feas,\dist)}\right) =: \text{Ex Ante Free-Order PI}(\feas).
    \end{align*} 
\end{theorem}
\begin{proof}
We have already argued one direction, namely that the worst case ratio of one-sided semi-online algorithms against the ex ante benchmark is no less than that of free-order online algorithms. For the other direction, let $\dist^*$ be the worst-case distribution for the one-sided semi-online versus ex ante gap. By mirroring the proof of~\Cref{lemma:expost-vs-exante}, we can assume without loss that $\dist^*$ is a Bernoulli product distribution, i.e. that each $X_i$ either takes value $0$ or $v_i$ with some probabilities. Indeed, if $q^*_i$ denote the selection probabilities of $\exante(\feas,\dist^*)$, then by having each element $X_i$ take value $0$ with probability $1-q^*_i$ and value $v_i=F_{\disti^*_i}(q^*_i)$ otherwise, the value of the ex ante optimal remains unchanged whereas the value of any one-sided semi-online algorithm can only decrease. Therefore, the only meaningful threshold for the optimal one-sided semi-online algorithm under $\dist^*$ when probing $X_i$ is $v_i$, and we can also assume without loss that if $X_i=0$ it will be permanently rejected as it doesn't contribute to the value or feasibility of the selection. Therefore, the optimal one-sided semi-online algorithm under the worst case distribution $\dist^*$ can be assumed to be a free-order online algorithm, from which the missing direction follows. 
\end{proof}


\paragraph{One-Sided Semi-Online and Frugal 
Algorithms.}
We now connect one-sided semi-online algorithms to another important class -- greedy-style algorithms for {\em deterministic} selection problems. We use the terminology and the formal definition given by \cite{S17}:


\begin{definition}[Frugal Algorithms]
   A deterministic selection algorithm $S:\R^n\mapsto \feas$ is called frugal if there exists an index function $g(i,S,y):[n]\times 2^{[n]} \times \R\mapsto\R$ that is non-decreasing on $y$, such that for all $x\in\R^n$, $S(x)$ is constructed incrementally by getting initialized at $S=\emptyset$ and then sequentially including the coordinate $j^*=\argmax_{i\notin S, S\cup\{i\}\in\feas}g(i,S,x_i)$ or halting. We say that the algorithm is $\beta$-frugal for some $\beta\leq 1$ if for all $x\in\R^n$ we have $\sum_{i\in S(x)}x_i \geq \beta\cdot\max_{S\in\feas}\sum_{i\in S}x_i$.
\end{definition}

Frugal approximations describe the ability of greedy criteria to approximate the optimal solution in a selection problem where elements have fixed values. In the following, we establish a close connection between frugal and one-sided semi-online algorithms; any frugal algorithm implies a one-sided semi-online algorithm for the same constraint, achieving the same approximation guarantee against the ex post optimum. The proof of the lemma follows \citeauthor{S17}'s argument relating the Price of Information for Pandora's Box problems to the frugal approximation and is deferred to the appendix.

\begin{lemma}\label{lemma:frugal-vs-semionline}
    Let $\feas$ be any constraint that admits an efficient $\beta$-frugal algorithm. Then, there exists a computationally-efficient one-sided semi-online algorithm for BCS over $\feas$ that $\beta$-approximates the ex post optimal for all distributions $\dist$.
\end{lemma}

Due to the connections that we established in~\Cref{lemma:expost-vs-exante} and~\Cref{lemma:frugal-vs-semionline}, from now on we will refer to the worst case gap between one-sided semi-online algorithms and the ex post (resp. ex ante) optimal benchmark as the \textbf{frugality} (resp. \textbf{frugal correlation gap}) of the constraint $\feas$. We use $\frugap(\feas)$ and $\frucorgap(\feas)$ to denote these quantities. By these definitions and~\Cref{lemma:expost-vs-exante}, the following is immediate.
\begin{fact}
    For any constraint $\feas$, we have $\frucorgap(\feas)\geq\frugap(\feas)\cdot \corgap(\feas)$.
\end{fact}

\section{Relating CICS and BCS}\label{sec:bcs-to-cics}

We can now establish a formal connection between CICS and BCS. We first relate CICS over Markov chains to one-sided semi-online algorithms for BCS. We begin by establishing the following crucial lemma, relating policies for MC-CICS to one-sided semi-online algorithms for BCS.

\begin{lemma}\label{lemma:connect-bcs-to-mc-cics}
    Let $\inst=(\feas,\mset)$ by any MC-CICS instance and let $\dist=W_1\times \cdots \times W_n$ be the product distribution over the surrogate values of the Markov chains in $\mset$.
    Then, given any one-sided semi-online algorithm $\so$ for the BCS instance $(\feas,\dist)$, there exists a computationally-efficient policy $\pol_\so$ for $\inst$ that obtains expected utility \[\util_{\inst}(\pol_\so)=\val_{\feas,\dist}(\so).\]    
\end{lemma}
\begin{proof}
We construct the policy $\pol_\so$ by simulating the execution of the one-sided semi-online algorithm $\so$ on the BCS instance $(\feas,\dist)$; whenever $\so$ probes a stochastic element $W_i$ with threshold $\tau_i$, our policy $\pol_\so$ advances the $i$-th Markov chain as long as the grade of its current state $s_i$ satisfies $g_i(s_i)\geq \tau_i$. If this advancement ends in a terminal state $t_i$ through a trajectory $\traj_i$ with $w_i(\traj_i)\geq \tau_i$ then we provide $\so$ with the feedback that $W_i=w_i(\traj_i)$ (and therefore has to be accepted); and otherwise we provide $\so$ with the feedback that $W_i < \tau_i$. In any case, we continue mirroring the execution of $\so$ through the above feedback mechanism until $\so$ eventually halts; at this point $\pol_\so$ accepts all the Markov chains corresponding to elements that $\so$ accepted throughout its execution and terminates.

Note that by definition, sampling a trajectory $\traj_i$ via a random walk in $\mc_i$ and returning $w_i(\traj_i)$ is equivalent to sampling from the surrogate value $W_i$. Furthermore, the index of a trajectory is defined as the minimum grade across all the states that it visits and therefore can only decrease throughout its evolution;  therefore, our feedback mechanism for the semi-online algorithm $\so$ is equivalent to running $\so$ on instance $(\feas,\dist)$. This implies that whenever nature realizes terminal trajectories $\traj_1,\cdots, \traj_n$ in the Markov chains, our policy $\pol_\so$ will accept the set of chains $S_\so(w)$ corresponding to the elements that $\so$ accepts whenever $w=(w_1(\traj_1),\cdots , w_n(\traj_n))$ gets realized. Finally, $\pol_\so$ trivially satisfies promise of payment; whenever a Markov chain is advanced, it will continue to do so until the index either decreases (and falls below the corresponding threshold) or it gets accepted. Combining everything, we obtain that
\[\probb{\traj_i \text{ realized in }\mc_i \text{ and }\pol_\so\text{ accepts }t(\traj_i)} = \probb{\so\text{ accepts }w_i(\traj_i)}\]
and therefore, by~\Cref{lemma:amortization-bound}, the utility of $\pol_\so$ is precisely
\begin{align*}
    \util_\inst(\pol_\so) =  \sum_{i=1}^n\sum_{\traj_i\in \trajset(\mc_i)} \probb{\so\text{ accepts }w_i(\traj_i)}\cdot w_i(\traj_i) = \expect{w\sim\dist}{\sum_{i\in S_\so(w)}w_i}
\end{align*}
from which the lemma follows.
\end{proof}

Using~\Cref{lemma:connect-bcs-to-mc-cics}, we can easily show that one-sided semi-online algorithms for BCS can be converted into efficient policies for MC-CICS without loss in performance.

\begin{theorem}\label{lemma:semi-online-implies-mc-cics}
    Fix any selection constraint $\feas$ and let $\so$ be any computationally-efficient one-sided semi-online algorithm that $\alpha$-approximates $\expost(\feas,\dist)$ for all distributions $\dist$. Then there exists a computationally-efficient policy $\pol_\so$ that $\alpha$-approximates $\opt(\inst)$ for all instances $\inst=(\feas,\mset)$ of MC-CICS.
\end{theorem}
\begin{proof}
    Let $\inst=(\feas,\mset)$ by any MC-CICS instance and let $\dist=W_1\times \cdots \times W_n$ be the product distribution over the surrogate values of the underlying Markov chains. Also, let $\so$ be any one-sided semi-online algorithm for $\feas$ that satisfies the conditions of the theorem; therefore, if $S_\so(w)\in \feas$ denotes the set of elements that $\so$ accepts when the stochastic values are realized to $w\in\R^n$, we have 
\begin{equation}\label{eq:mc-cics-red-1}
  \val_{\feas,\dist}(\so)=\expect{w\sim\dist}{\sum_{i\in S_\so(w)}w_i}\geq \alpha \cdot \expect{w\sim\dist}{\max_{S\in\feas}\sum_{i\in S}w_i}. 
\end{equation}

Now, let $\pol^*$ be the optimal policy for the MC-CICS instance $\inst$ and note that $\pol^*$ may or may not satisfy promise of payment; in any case, by using the upper bound of~\Cref{lemma:amortization-bound} we have that
\begin{equation}\label{eq:mc-cics-red-2}
  \opt(\inst) \leq \sum_{i=1}^n\sum_{\traj_i\in \trajset(\mc_i)} \probb{\traj_i \text{ realized in }\mc_i \text{ and }\pol^*\text{ accepts }t(\traj_i)}\cdot w_i(\traj_i) \leq \expect{w\sim\dist}{\max_{S\in\feas}\sum_{i\in S}w_i}
\end{equation}
where the last inequality follows from the fact that the set of Markov chains accepted by $\pol^*$ lies in $\feas$ with probability $1$, and therefore the best case for $\pol^*$ is to always accept the feasible set of realized surrogate values that achieve maximum total value. The proof is completed by combining inequalities~\eqref{eq:mc-cics-red-1} and~\eqref{eq:mc-cics-red-2} with the policy $\pol_\so$ of~\Cref{lemma:connect-bcs-to-mc-cics}, as the computational efficiency of $\pol_\so$ is guaranteed.
\end{proof}
As a corollary, \Cref{lemma:semi-online-implies-mc-cics} and \Cref{lemma:frugal-vs-semionline} together imply that frugal algorithms can be used to construct efficient policies for MC-CICS, a result that was previously proved by \cite{GJSS19}. Our work strengthens this result by constructing good policies from any one-sided semi-online algorithm, not just those corresponding to frugal algorithms. In the following, we further argue that the converse statement is also true.

\begin{theorem}
\label{thm:mc-cics-converse}
    Fix any constraint $\feas$ and any $\epsilon>0$. Given a computationally-efficient policy $\pol$ for MC-CICS that $\alpha$-approximates $\opt(\inst)$ for all instances $\inst=(\feas,\mset)$ and satisfies promise of payment, there exists a computationally efficient one-sided semi-online algorithm $\so_\pol$ that $(\alpha-\epsilon)$-approximates $\expost(\feas,\dist)$ for all distributions $\dist$.
\end{theorem}
\begin{proof}
We exploit the same connections between a MC-CICS instance and a BCS instance over the surrogate values of the underlying Markov chains as in the proof~\Cref{lemma:connect-bcs-to-mc-cics}, with one important distinction. After selecting an index $i$ and a threshold $\tau_i$, a one-sided semi-online algorithm must accept it if $W_i\geq\tau_i$ and learns that $W_i<\tau_i$ otherwise. On the other hand, whenever a MC-CICS policy that satisfies promise of payment advances a state of grade $\tau_i$, it must continue advancing until it encounters a state of grade $g<\tau_i$, at which point it has learned that the index of the realized trajectory will be at most $g$. Therefore, MC-CICS policies potentially obtain more information compared to one-sided semi-online algorithms whenever they do not accept an element. We overcome this asymmetry by constructing specific Markov chains whose surrogate values mirror a target distribution while ensuring that the above scenario cannot happen.

Formally, let $\disti$ be any distribution over real numbers and let $\mathrm{supp}(\disti)=\{w_1,\cdots , w_k\}$ where the values $w_i$ are in increasing order $w_1<w_2<\cdots <w_k$, and let each $w_i$ be realized with some probability $p_i$. Now, let $\mc(\disti)$ be the tree-like Markov chain with state space $S=\{s_2,\cdots , s_k, t_1, t_2,\cdots , t_k\}$, starting state $s_k$ and terminal states $T=\{t_1,\cdots,t_k\}$; we use $s_1=t_1$ for convenience. Each non-terminal state $s_i$ for $i\geq 2$ has a single action of cost $c=\epsilon$ for some $\epsilon>0$ that either transitions to $s_{i-1}$ or to $t_i$, with probabilities $1-q_i$ and $q_i$ respectively, where $q_i=p_i/\sum_{j\leq i}p_j$. Finally, the values at the terminal states are defined as $v_i(t_i)=w_i + \epsilon/q_i$. Using the amortization framework of~\cite{CCHS24}, according to which the grades and surrogate values of a tree-like Markov chain can be defined, it is trivial to verify that the grade of each state $s_i$ is precisely $g(s_i)=w_i$ by the ordering of the $w_i$'s and the hierarchical structure of our constructed Markov chain $\mc(\disti)$. This implies that the distributions of the surrogate value $W$ for $\mc(\disti)$ is precisely $\disti$.

Let $(\feas,\dist)$ by any BCS instance and let $\pol$ be any policy for the MC-CICS instance $\inst=(\feas,\mset)$ where $\mc_i=\mc(\disti_i)$ are constructed as above. Now consider the one-sided semi-online algorithm $\so_\pol$ that simulates the execution of $\pol$ in the following way: whenever $\pol$ advances $\mc_i$ from a state of grade $g_i$, the algorithm $\so_\pol$ will probe the $i$-th element with threshold $\tau_i=g_i$. Assuming that $\pol$ satisfies promise of payment, if the Markov chain transitions to a terminal state the index will not decrease; therefore $\pol$ will be forced to accept this terminal much like the one-sided semi-online algorithm $\so$. Furthermore, whenever this doesn't happen, we transition to a state of lesser index; by construction this index corresponds to the second maximum value in the support of $W_i$ and therefore $\pol$ only learns that $W_i < g_i$, which is the same feedback as $\so_\pol$ obtains. From this, we obtain that
\[\val_{\feas,\dist}(\so_\alg)=\expect{w\sim\dist}{\sum_{i\in \so_\alg(w)}w_i} = \sum_{i=1}^n\sum_{t_i\in T_i} \probb{\pol_\so\text{ accepts }t_i}\cdot w_i(t_i) = \util(\pol)\]
where the last equality follows from promise of payment and~\Cref{lemma:amortization-bound}.

Finally, by the assumption on the approximation ratio of $\pol$ we have that $\util(\pol)\geq \alpha \cdot \opt(\inst)$. To complete the proof we need to show that $\opt(\inst)\geq \expect{w\sim\dist}{\max_{S\in\feas}\sum_{i\in S}w_i}$ i.e. that the optimal policy in $\inst$ matches the performance of the ex post optimal. However, this is not necessarily true; in fact, we know that the opposite direction of the inequality holds from~\Cref{lemma:amortization-bound}. This would hold with equality if $\epsilon=0$ (since all the Markov chains could be advanced freely without loss) but the promise of payment would be trivially satisfied by fully advancing all the Markov chains. Instead, by setting $\epsilon$ to be arbitrarily small, we have that $\opt(\inst)$ gets arbitrarily close to the ex post optimal benchmark, since we can fully advance all the Markov chains to their terminal states up to an arbitrarily small cost. Lastly, we note that our entire construction if efficient (with respect to the sizes of the distributions).
\color{black}
\end{proof}



Next, we prove a counterpart of~\Cref{lemma:semi-online-implies-mc-cics} for general CICS. In particular, we show that approximations of one-sided semi-online algorithms against the {\em ex ante} optimum extend to approximations for the CICS via committing algorithms.

\begin{theorem}\label{lemma:cics-reduction}
    Fix any feasibility constraint $\feas$ and let $\so$ be a one-sided semi-online algorithm that $\alpha$-approximates $\exante(\feas,\dist)$ for all distributions $\dist$. Then there exists a committing policy $\pol_\so$ that $\alpha$-approximates $\opt(\inst)$ for all CICS instances $\inst=(\feas,\mset)$.
\end{theorem}
\begin{proof}
Let $\inst=(\feas,\mset)$ be any CICS instance. We begin by using~\Cref{lemma:exante-compute} to efficiently determine a commitment $\comms\in\comset(\inst)$ such that $\exanteopt(\inst) = \exanteopt(\inst^\comms)$. By definition, we have that $\opt(\inst)\leq \exanteopt(\inst)$. Furthermore, as as a direct consequence of~\Cref{lemma:amortization-bound}, we have that $\exanteopt(\inst^\comms)\leq \exanteopt(\feas,\dist^\comms)$ where $\dist^\comms$ denotes the product distribution over the surrogate values of the Markov chains in $\inst^\comms$. Therefore, we have efficiently determined a commitment $\comms\in\comset(\inst)$ such that $\opt(\inst)\leq \exanteopt(\feas,\dist^\comms)$.

We now use~\Cref{lemma:connect-bcs-to-mc-cics}, instantiated over $\inst^\comms$ and a one-sided semi-online algorithm $\so$ for $\feas$; note that the resulting MC-CICS policy $\pol_\so$ for $\inst^\comms$ is also a committing policy for $\inst$. Furthermore, by the guarantees of~\Cref{lemma:connect-bcs-to-mc-cics} and our upper bound on $\opt(\inst)$, we obtain that this policy will satisfy an approximation ratio of at least
$\val_{\feas,\dist^\comms}(\so)/\exanteopt(\feas,\dist^\comms)$ from which the proof follows assuming that $\so$ satisfies the required guarantees with respect to the ex ante optimum.

\color{black}
\end{proof}

\noindent Finally, by combining~\Cref{lemma:cics-reduction} with~\Cref{lemma:semi-online-and-eafopi}, the following is immediate.
\begin{corollary}\label{cor:eafopi-bounds-comgap}
    For any CICS instance $\inst=(\feas,\mset)$, $\comgap(\inst)\geq \text{Ex Ante Free-Order PI}(\feas)$.
\end{corollary}

\section{Computing an Ex Ante Optimal Policy for CICS}\label{sec:cics-exante}

In this section, we consider the ex ante relaxation of CICS and prove~\Cref{lemma:exante-compute}. Fix any CICS instance $\inst$ and let $\pol^*$ be any optimal ex ante feasible policy for $\inst$. We use $q_i$ to denote the probability that $\pol^*$ accepts the $i$-th MDP upon terminating. By definition of ex ante feasibility, we have that $q\in\polytope(\feas)$. As~\cite{BLW25} observe, subject to accepting the $i$-th MDP with probability $q_i$, the policy's actions within the MDP can be determined independently of its actions and states in other MDPs due to the independent randomness. In other words, the policy can independently optimize its net reward from each MDP $\mc_i$ subject to the ex ante acceptance probability $q_i$. This is the key property that enables efficient computation of ex ante optimal policies in contrast to computing optimal ex post feasible policies, which is known to be a computationally hard problem even in special cases~\citep{FLL22}. We formally capture this notion of independent optimization via the following concept of local games.



\begin{definition}[Local Games]
    Let $\mc$ be any MDP. We consider the local game corresponding to an instance $\inst$ of CICS comprised only by $\mc$ and the (single-selection) constraint $\feas=\{\emptyset,\mc\}$. For any acceptance constraint $q\in [0,1]$, we use $h_\mc(q)$ to denote the maximum utility obtained across all policies $\pol$ for $\inst$ that accept a terminal state from $\mc$ with probability at most $q$, that is, $\probb{\mc\in S_\pol} \leq q$. 
\end{definition}

Observe that by associating accepting the empty set in the local game $(\mc_i,q_i)$ to a halting action, $h_{\mc_i}(q_i)$ captures precisely the (optimal) utility that an ex ante feasible policy can extract from the $i$-th MDP subject to accepting it with probability at most $q_i$. This leads to the following observation:
\begin{fact}
    Let $\inst=(\feas,\mset)$ be any CICS instance. Then, the utility of the optimal ex ante feasible policy for $\inst$ is 
\[\exanteopt(\inst):=\max_{q\in \polytope(\feas)}\sum_{i=1}^n h_{\mc_i}(q_i).\]
\end{fact}

Our objective in this section will be to find an optimal solution $q^*\in\polytope(\feas)$ for the above program, as well as the local policies that achieve the optimal utilities $h_{\mc_i}(q^*_i)$ for all indices $i\in [n]$. It is not hard to see that each function $h_{\mc_i}(q_i)$ is concave\footnote{For $q,q',\lambda\in [0,1]$, the local policy that is equal to a maximizer of $h_\mc(q)$ with probability $\lambda$ and to a maximizer of $h_\mc(q')$ otherwise will only accept a terminal state with probability at most $\lambda q + (1-\lambda)q'$ and therefore $h_\mc(\lambda q + (1-\lambda)q') \geq \lambda h_\mc(q) + (1-\lambda)h_\mc(q')$.} and therefore the same is true for the function $h(q)=\sum_{i=1}^nh_{\mc_i}(q_i)$. Furthermore, we will be assuming throughout that the combinatorial constraint $\feas$ generates a polytope $\polytope(\feas)$ that admits an efficient separation oracle; this is the case for all the constraints that we study in this work. Therefore, efficient computation of $\exanteopt(\inst)$ reduces to efficient computation of $h_{\mc_i}(q_i)$. 

We also note that the local policies achieving the optimal utilities $h_{\mc_i}(q_i)$ could potentially take actions in $\mc_i$ that depend on the entire trajectory of $\mc_i$ up until this point. In contrast, we wish to establish the existence of optimal \textit{memoryless policies} whose actions only depend on the current state, as such policies correspond to valid commitments for the underlying MDP. Shifting our attention to these sub-problems, our main contribution is the following:

\begin{lemma}\label{thm:quantile-computation}
    Fix any local game $(\mc,q)$. Then, we can compute $h_\mc(q)$ as well as a commitment $\pi$ for $\mc$ such that $h_{\mc}(q)=h_{\mc^\pi}(q)$ in time $\mathcal{O}\big(\mathrm{poly}(|\mc|)\big)$. 
\end{lemma}
From the above discussion, the efficient computation of $h_\mc(q)$ for a local game directly implies efficient computations of $\exanteopt(\inst)$ for a CICS instance $\inst$. Furthermore, since the local policies applied to each individual MDP $\mc_i$ are committing policies $\pi_i$, we also obtain that $\exanteopt(\inst)=\exanteopt(\inst^\comms)$ where $\comms=(\pi_1,\ldots, \pi_n)$ and the proof of~\Cref{lemma:exante-compute} follows. The remainder of this section is dedicated towards the proof of~\Cref{thm:quantile-computation}.

\paragraph{Comparison to Previous Work.} We note that up until~\Cref{thm:quantile-computation}, our approach follows precisely the same steps as in~\cite{BLW25}. Their counterpart to~\Cref{thm:quantile-computation} shows that there exists an FPTAS for DAG-MDPs that is obtained by estimating each $h_{\mc_i}(q_i)$ in an inductive way. In contrast, our result also applies to general MDPs that might contain directed cycles, and also provides an exact solution. Our result is obtained by formulating the computation of $h_{\mc_i}(q_i)$ as a constrained MDP, and using standard techniques from this literature, namely an occupation-measure LP, to solve it. We refer the reader to the manuscript of~\cite{Altman1999} for more details on constrained MDPs and related techniques. 

\subsection{Computing Optimal Solutions for Local Games}

From now on, fix any local game $(\mc,q)$. By a slight abuse of notation, we augment the state space $\Omega$ of $\mc$ to include an absorbing state $s_\mathrm{abs}$ from which no more actions can be taken. We then augment the action sets and transition probabilities of all states in the following manner:
\begin{itemize}
    \item For any non-terminal state $s\in \Omega\setminus T$, we re-label the action set from $A(s)$ to $A_{\mathrm{adv}}(s)$ and use $A(s)=A_{\mathrm{adv}}(s)\cup\{\textsf{halt}\}$; in other words, every non-terminal state is characterized by a (non-empty) set of advancing actions $a\in A_{\mathrm{adv}}(s)$ (each at a non-negative cost $c(a)\geq 0$ leading to a random state $s'$ with probability $\Pi(s,a,s')$), and a halting action of cost $c(\textsf{halt})=0$, leading to $s_\mathrm{abs}$ with probability $1$.

    \item For any terminal state $s\in T$, we denote $A(s)=\{\textsf{accept},\textsf{halt}\}$. Both actions have a cost of $0$ and lead to the absorbing state $s_\mathrm{abs}$ with probability $1$. The accept action also produces a reward $v(s)$.
\end{itemize}

Note that any policy for $\mc$ can be described as a protocol that given the current history and state, specifies (potentially randomly) which action should be taken, until the absorbing state is reached (corresponding to having either halted or accepted a terminal state). To account for the utility of this protocol, we define
\[u(s,a) := \begin{cases}
    - c(a) &\text{ if } s\in \Omega\setminus T \text{ and } a\in A_{\mathrm{adv}}(s)\\
    \; \; v(s) &\text{ if } s\in T \text{ and } a=\textsf{accept} \\
    \; \; \; 0 &\text{ if } s\in \Omega \text{ and } a=\textsf{halt} \\
\end{cases}.\]
Equipped with the above, we can now formulate a linear program for the local game $(\mc,q)$:

\begin{equation}
\label{eq:LP}
\begin{aligned}
\text{maximize } & \sum_{s \in \Omega} \sum_{a \in A(s)} z(s,a)\, u(s,a) \\
\text{subject to } &
\sum_{a \in A(s)} z(s,a)
-
\sum_{s' \in \Omega} \sum_{a' \in A(s')} z(s',a') \cdot \Pi(s',a',s)
= \mathbbm{1}[s = \sigma] \quad \forall s \in \Omega,  \\
& \sum_{s \in T} z(s,\textsf{accept}) \le q, \\
& z(s,a) \ge 0 \quad \forall s\in \Omega,a\in A(s).
\end{aligned}
\end{equation}

The above formulation is known as an occupation-measure LP. The variables $z(s,a)$ denote the expected number of times a policy takes action $a\in A(s)$ on state $s\in\Omega$ and therefore the objective precisely captures the utility of this policy. The first constraint ensures that the trajectory of the policy starts on the root state $\sigma$ and that the law of conditional expectation is satisfied in every state. The second constraint ensures that the trajectory ends by accepting a terminal state with probability at most $q$, and therefore the policy is feasible for the local game $(\mc,q)$. From these connections, the following is immediate and the proof is omitted:
\begin{fact}
    Let $V^*$ be the value of the optimal solution for LP~\eqref{eq:LP}. Then, $h_\mc(q)\leq V^*$.
\end{fact}

We also note that the size of the above LP is linear to the size of the MDP, and therefore we have efficient access to an optimal solution $z^*$. The following lemma establishes that this optimal solution can be transformed into a memoryless policy for $(\mc,q)$ without any loss in utility. Together, these imply the proof of~\Cref{thm:quantile-computation}.

\begin{lemma}
    There exists a memoryless policy for $(\mc,q)$ that achieves utility $V^*$.
\end{lemma}
\begin{proof}
    For each state $s\in \Omega$, let $N^*(s):=\sum_{a\in A(s)}z^*(s,a)$ denote the expected number of visits in the optimal solution. We define the memoryless policy $\phi$ that upon reaching state $s\in \Omega$, takes action $a\in A(s)$ with probability $\phi(s,a)=z^*(s,a)/N^*(s)$; if $N^*(s)=0$ for some state $s$, we instead set $\phi(s,\textsf{halt})=1$. Note that since positive utility can only be obtained by accepting terminal states, after which we transition to an absorbing state and the game ends, any optimal solution to the above LP will produce measures $N^*(s)$ that are finite (i.e. the expected number of steps to end the game is finite). Therefore, we have indeed defined a valid memoryless policy.\footnote{To transform this memoryless policy into the commitment specified in~\Cref{thm:quantile-computation}, for each state $s\in\Omega\setminus T$ and each $a\in A_{\mathrm{adv}}(s)$ we set $\pi(s,a)=\phi(s,a)/(1-\phi(s,\mathrm{halt}))$ i.e. we consider the probability that $\phi$ takes action $a$ conditioned on not halting. Is $\phi(s,\mathrm{halt})=1$ for some state $s$ then this state will never be advanced and we can commit to any action without loss.} 
    
    Next, let $z_\phi$ denote the occupancy measure of policy $\phi$, that is, $z_\phi(s,a)$ is the expected number of times that policy $\phi$ takes action $a$ from state $s\in \Omega$. Furthermore, let $N_\phi(s)=\sum_{a\in A(s)}z_\phi(s,a)$ denote the expected number of times that state $s$ is visited by $\phi$. Note that by definition of $\phi$, we have $z_\phi(s,a)=N_\phi(s)\cdot \phi(s,a)$. Then, the utility achieved by $\phi$ can be written as
    \[\util (\phi):=\sum_{s\in\Omega} N_\phi(s) \sum_{a\in A(s)}\phi(s,a)u(s,a) = \sum_{s\in \Omega:N^*(s)\neq 0}\frac{N_\phi(s)}{N^*(s)}\sum_{a\in A(s)}z^*(s,a)u(s,a)\]
    with the second equality following from the fact that if $N^*(s)=0$ for some state $s$, then $z^*(s',a)=0$ for all states $s'\in \Omega$ for which $\Pi(s',a,s)\neq 0$ and therefore policy $\phi$ will never reach $s$. Furthermore, since accepting a terminal state results to transitioning to the absorbing state and the game ending, the probability that $\phi$ accepts a terminal state is given by
    \[\probb{\phi \text{ accepts a terminal state}} = \sum_{s\in T} N_\phi(s)\phi(s,\textsf{accept}) = \sum_{s\in T:N^*(s)\neq 0} \frac{N_\phi(s)}{N^*(s)}z^*(s,\textsf{accept}).\]

    Therefore, to prove the lemma, it suffices to show that $N^*(s)=N_\phi(s)$ for all states $s\in \Omega$. Since $z^*$ is feasible for the LP, we have that for all $s\in \Omega$,
    \begin{align*}
    \mathbbm{1}[s = \sigma]  &= \sum_{a \in A(s)} z^*(s,a) - \sum_{s' \in \Omega} \sum_{a' \in A(s')} z^*(s',a') \cdot \Pi(s',a',s) \\
    &= N^*(s) - \sum_{s' \in \Omega}N^*(s') \sum_{a' \in A(s')} \phi(s',a') \cdot \Pi(s',a',s) \\
    \end{align*}
    On the other hand, the memoryless policy $\phi$ for $\mc$ can be equivalently viewed as a random trajectory on the Markov chain $\mc_\phi$ where upon reaching a state $s$, a random action is sampled automatically from $\phi(s,\cdot)$ until the absorbing state is reached. Given that the starting state for $\phi$ is $\sigma$, the expected visit times $N_\phi(s)$ will therefore satisfy 
      \begin{align*}
    \mathbbm{1}[s = \sigma]  &=  N_\phi(s) - \sum_{s' \in \Omega}N_\phi(s') \sum_{a' \in A(s')} \phi(s',a') \cdot \Pi(s',a',s). 
    \end{align*}
    Therefore, $N_\phi$ and $N^*$ satisfy the exact same set of linear equations, whose solution corresponds to the expected visit times of a trajectory in a Markov chain. Assuming that the Markov chain is absorbing (i.e. an absorbing state is reachable from every state of the chain with non-zero probability), this solution is unique and therefore $N_\phi(s) = N^*(s)$ for all $s\in \Omega$ as desired. This is clearly true as every action in the MDP that doesn't halt or accept a terminal state (i.e. any action that doesn't reach the absorbing state) has a negative utility (equal to the action cost) and therefore $\phi$ never reaching the absorbing state would correspond to $z^*(s,\textsf{halt}) =0$ for all $s\in\Omega$ and $z^*(s,\textsf{accept})=0$ for all $s\in T$, implying that $V^*\leq 0$ -- in that case, we can immediately halt instead and still achieve a utility of $0$.
\end{proof}

\section{One-Sided Semi-Online Matching Is Query-Commit Matching}\label{sec:query-commit-semi-online}

In this section, we give additional results on the frugality of the bipartite and general matching constraints. Note that since any matching constraint is a $2$-system, it inherits a frugal gap of $\frac{1}{2}$ by~\Cref{table:res}. We improve on this bound by directly comparing the performance of any one-sided semi-online algorithm against the ex-post optimum matching. To leverage this, we turn our attention to the well-studied \emph{query-commit} model for stochastic weighted (bipartite) matching, and show that designing a good one-sided semi-online algorithm is equivalent to designing a good query-commit algorithm.

In the query-commit model, introduced by~\cite{chen2009approximating}, the algorithm is given a graph $G = (V, E)$, where every edge $e \in E$ has a weight $w_e \ge 0$ and exists independently with probability $p_e$. During its execution, the algorithm may (adaptively) query any edge to see if it exists or not. If it does, the algorithm must commit to including this edge in its output set. At the end, the algorithm must output a valid matching of the graph $G$. There have been many exciting developments in this model~\citep{gamlath2019beating,derakhshan2023beating,chen2025significance,huang2025edge,fu2021random,macrury2024random}, with the state-of-the-art competitive ratios being $0.659$ for bipartite matching~\citep{huang2025edge}, and $0.535$ for general matching~\citep{macrury2024random}.

In the context of the Bayesian Combinatorial Selection problem, a one-sided semi-online algorithm on an instance where every distribution is Bernoulli behaves exactly as that of a query-commit algorithm. Therefore, if an $\alpha$-competitive one-sided semi-online algorithm exists, then it can be trivially modified into an $\alpha$-competitive query-commit algorithm. We essentially prove the converse statement, showing that under the (bipartite) matching constraint, one-sided semi-online algorithms and query-commit algorithms are as powerful as each other.

\begin{theorem}
\label{thm:qc-equals-semi-online}
    Given an algorithm $\mathcal{A}$ that is $\alpha$-competitive for the query-commit matching model, for any $\eps > 0$, we can construct an $(\alpha - \eps)$-competitive one-sided semi-online algorithm $\mathcal{A'}$ for the stochastic matching problem.
\end{theorem}

Using previous results for query-commit matching, we get the following corollary.

\begin{corollary}
    For any $\eps > 0$, there exists a $(0.659 - \eps)$-competitive one-sided semi-online algorithm for BCS under a bipartite matching constraint, and a $(0.535 - \eps)$-competitive one-sided semi-online algorithm for BCS under a general matching constraint.
\end{corollary}

\subsection{Reduction under a Well-Behaved Assumption}

To begin with, we first prove the theorem under a certain \emph{well-behaved} assumption of the algorithm $\mathcal{A}$. We define it as follows.

\begin{definition}[Well-Behaved Query-Commit Algorithm]
    An algorithm $\mathcal{A}$ for the query-commit model is \emph{well-behaved} iff for any instance ($G, \mathbf{w}, \mathbf{p}$), and for any pair of parallel edges $e, e' \in E$ such that $w_e > w_{e'}$, $\mathcal{A}$ must query $e$ before $e'$ on any trajectory of the algorithm.
\end{definition}

The well-behaved assumption states that whenever the algorithm attempts to match two vertices in $G$ that are connected via multiple parallel edges, it does so by probing the edge of maximum weight. The following proves the equivalence between query-commit and one-sided semi-online algorithms under this assumption. 
\begin{lemma}
\label{lem:well-behaved-qc-equals-semi-online}
    Fix any one-sided semi-online instance $(G, \D)$. For any well-behaved algorithm $\mathcal{A}$ that is $\alpha$-competitive for the query-commit model, we can construct a query-commit instance $(G', \mathbf{p}', \mathbf{w}')$ and a one-sided semi-online algorithm $\mathcal{A'}$ such that $\expost(G, \D) = \expost(G', \mathbf{p}, \mathbf{w})$    and $\mathcal{A'}(G, \D) = \mathcal{A}(G', \mathbf{p}, \mathbf{w})$.
\end{lemma}

\begin{proof}
Consider~\Cref{alg:split-support}.
    
\begin{algorithm}[tbh]
\caption{$\textsc{SplitSupport}_\mathcal{A}(G = (V, E), \D)$}
\begin{algorithmic}[1]
\label{alg:split-support}
\Require $\mathcal{A}$ is a well-behaved query-commit algorithm. $(G, \D)$ is a one-sided semi-online stochastic matching instance.
\State Let $G' = (V, E' \gets \emptyset), \mathbf{p}', \mathbf{w}'$ be a new query-commit instance.
\For{$e \in E$}
    \For{$s \in \operatorname{support}(D_e) \cap \R_{> 0}$}
        \State Add a copy $e'_s$ of edge $e$ into $E'$, with $w'_{e'_s} = s$ and $p'_{e'_s} = 1 - \frac{\Pr_{w \sim D_e}[w < s]}{\Pr_{w \sim D_e}[w \le s]}$.
    \EndFor
\EndFor
\State Simulate $\mathcal{A}(G', \mathbf{p}', \mathbf{w}')$.
\While{$\mathcal{A}$ queries some $e'_s$ on $G'$}
    \State Query edge $e$ at threshold $s$ on $G$.
    \State Forward the output of this query to $\mathcal{A}$.
\EndWhile
\end{algorithmic}
\end{algorithm}

We first prove that $\expost(G, \D) = \expost(G', \mathbf{p}, \mathbf{w})$. This comes from the following two observations.
\begin{itemize}
    \item Should the optimal ex-post algorithm choose a copy of the edge $e \in G$ in $G'$, it must choose the available copy with maximum weight, and
    \item For every edge $e \in G$, the distribution of ``maximum weight of an available copy of $e$ in $G'$'' is exactly $D_e$: for every $t \ge 0$,
    \begin{align*}
        &\Pr[\text{maximum weight of an available copy of $e$ in $G'$} \le t] \\
        &= \prod_{\text{all copies $e'_s$ with $s > t$}} (1 - \Pr[\text{$e'_s$ is available}]) \\
        &= \prod_{\text{all copies $e'_s$ with $s > t$}} \frac{\Pr_{s \sim D_e}[w < s]}{\Pr_{w \sim D_e}[w \le s]} \\
        &= \frac{\Pr_{w \sim D_e}[w \le t]}{\Pr_{w \sim D_e}[w \le \max\operatorname{support}(D_e)]} = \Pr_{w \sim D_e}[w \le t]
    \end{align*}
    where the third equality is because every distinct support of $D_e$ above $t$ is present as the weight of some copy, so the product telescopes.
\end{itemize}

Now, we show that $\textsc{SplitSupport}_\mathcal{A}(G, \D) = \mathcal{A}(G', \mathbf{p}, \mathbf{w})$. This follows from the following two observations.
\begin{itemize}
    \item For all edge $e \in E$, $\textsc{SplitSupport}_\mathcal{A}$ selects the edge $e$ under weight $w$ iff $\mathcal{A}$ selects the copy $e'_w$ of $e$. This is because whenever $\mathcal{A}$ queries $e'_w$ and selects it, it must mean that $\mathcal{A}$ has queried all copies of $e$ that have weight larger than $w$ and they were all inactive (due to $\mathcal{A}$ being well-behaved). This means that on $\textsc{SplitSupport}_\mathcal{A}$, it has queried on edge $e$ all possible thresholds above $w$ and failed, but succeeds when querying at threshold $w$, implying $\textsc{SplitSupport}_\mathcal{A}$ selects edge $e$ at weight $w$.
    \item Furthermore, when a copy $e'_s$ with probability $p'_{e'_s}$ is queried in $\mathcal{A}$, by the observation above, $\textsc{SplitSupport}_\mathcal{A}$ knows that edge $e$'s weight cannot be larger than $s$. Therefore, the conditional probability of the query succeeding in $\textsc{SplitSupport}_\mathcal{A}$ is exactly
    \[
        \Pr_{w \sim D_{e}}[w \ge s \mid w \le s] = 1 - \Pr_{w \sim D_{e}}[w < s \mid w \le s] = 1 - \frac{\Pr_{w \sim D_e}[w < s]}{\Pr_{w \sim D_e}[w \le s]} = p'_{e'_s}.
    \]
\end{itemize}
\end{proof}

\subsection{Enforcing the Well-Behaved Assumption under an Instance Assumption}

We next prove that under certain assumptions about the query-commit instance, assuming that $\mathcal{A}$ is well-behaved is without loss of generality.

\begin{lemma}
\label{lem:eq-prob-make-well-behaved}
    Given any algorithm $\mathcal{A}$ for the query-commit model, and an instance $(G, \mathbf{p}, \mathbf{w})$ where every edge $e \in G$ has the same probability of appearance, we can construct a well-behaved algorithm $\mathcal{A}'$ such that $\mathcal{A'}(G, \mathbf{p}, \mathbf{w}) \ge \mathcal{A}(G, \mathbf{p}, \mathbf{w})$.
\end{lemma}

To achieve this goal, we consider the following algorithm that turns an algorithm under such an instance into a well-behaved one.

\begin{algorithm}[tbh]
\caption{$\textsc{MakeBehaved}_\mathcal{A}(G = (V, E), \mathbf{p}, \mathbf{w})$}
\begin{algorithmic}[1]
\label{alg:make-behaved}
\Require $\mathcal{A}$ is any query-commit algorithm. $(G, \mathbf{p}, \mathbf{w})$ is a query-commit instance where for all $e \in E$, $p_e = p$ for some $p \in [0, 1]$.
\State Simulate $\mathcal{A}(G, \mathbf{p}, \mathbf{w})$.
\While{$\mathcal{A}$ queries some $e$ on $G'$}
    \State Query a parallel edge of $e$ in $G$ that has maximum weight among those parallel edges that are unqueried, breaking ties arbitrarily.
    \State Forward the output of the query to $\mathcal{A}$.
\EndWhile
\end{algorithmic}
\end{algorithm}

This stems from the intuition that querying a large-weight parallel edge earlier rather than later avoids committing to a low-weight parallel edge before observing the large-weight one. Furthermore, since edges exist with equal probability, modifying which parallel copy of an edge to query does not change the ``risk'' of committing to such a copy. The detailed proof is deferred to~\Cref{app:proofs}.

The previous lemma can be generalized to the following lemma.

\begin{lemma}
\label{lem:powers-of-eps-make-behaved}
    Given any algorithm $\mathcal{A}$ for the query-commit model, and an instance $(G, \mathbf{p}, \mathbf{w})$, where there exists some $\eps \in (0, 1)$ such that for every edge $e \in E$, $p_e = 1 - (1 - \eps)^{k_e}$ for some $k_e \in \mathbb{Z}_{\ge 0}$, we can construct a well-behaved algorithm $\mathcal{A}'$ such that $\mathcal{A'}(G, \mathbf{p}, \mathbf{w}) \ge \mathcal{A}(G, \mathbf{p}, \mathbf{w})$.
\end{lemma}

At a high level, such an instance can be reduced to an equiprobable instance by duplicating the edge $e$ exactly $k_e$ times, where each copy has an $\eps$ probability of appearance. Observe that the probability that at least one copy of $e$ appears is exactly $p_e = 1 - (1 - \eps)^{k_e}$, so these copies exactly mimic the appearance probability of the edge $e$. For an algorithm that queries a copy of $e$, we simulate the output of this query by querying $e$ with some carefully chosen probability, such that the conditional acceptance probability is exactly $\eps$. The full proof can be found in~\Cref{app:proofs}.

\subsection{Removing the Instance Assumption}

We have shown via~\Cref{lem:well-behaved-qc-equals-semi-online} that given a well-behaved query-commit algorithm, a competitive one-sided semi-online algorithm can be constructed via splitting the support of each edge distribution into many parallel take-it-or-leave-it copies of the same edge. We have also shown via~\Cref{lem:powers-of-eps-make-behaved} that if the reduced query-commit instance satisfies some $\eps$-grid condition, then in fact every query-commit algorithm can be turned well-behaved. This subsection shows that by slightly modifying the original query-commit instance such that the reduced query-commit instance satisfies $\eps$-grid the condition of~\Cref{lem:powers-of-eps-make-behaved}, we only incur a small loss to the competitive ratio, thus finalizing the argument for~\Cref{thm:qc-equals-semi-online}. Our main technical lemma is the following.

\begin{lemma}
Let $G = (V, E)$ be any graph, and $\D$ and $\D' = \D_{-e} \times D'_e$ be two product distributions over the edges $E$, such that only the marginals for the edge $e \in E$ are different in both distributions. Assume that all marginals of $\D$ and $\D'$ have supports in $[0, \Delta]$. Then we have
\[|\expost(G, \D) - \expost(G, \D')| \le |V| \cdot \Delta \cdot d_\mathrm{TV}(D_e, D'_e), \]
and furthermore, for any (one-sided semi-online) decision tree $\mathcal{A}$, we have
\[|\mathcal{A}(G, \D) - \mathcal{A}(G, \D')| \le |V| \cdot \Delta \cdot d_\mathrm{TV}(D_e, D'_e). \]
\end{lemma}

\begin{proof}
    The proof follows directly from the fact that there exists a coupling between $D_e$ and $D'_e$ such that the edge $e$ is initialized with different values with probability exactly $d_\mathrm{TV}(D_e, D'_e)$. Such coupling can be extended to the product distributions $\D$ and $\D'$ such that the initialized instances are different with probability exactly $d_\mathrm{TV}(D_e, D'_e)$. When the instances are the same, the ex-post optimums are the same. When the instances are different, we bound the difference between the ex-post optimums by the difference between the maximum achievable ex-post optimal (which is $\Delta \cdot |V|$) and the minimum achievable ex-post optimal (which is $0$). This shows that $|\expost(G, \D) - \expost(G, \D')| \le |V| \cdot \Delta \cdot d_\mathrm{TV}(D_e, D'_e)$ as claimed. The second statement can be derived similarly on the execution of the decision tree $\mathcal{A}$.
\end{proof}

By applying the same lemma $|E|$ times for each pair of marginals on the same edge, we arrive at the following corollary.

\begin{corollary}
\label{cor:tv-distance}
Let $G = (V, E)$ be any graph, and $\D$ and $\D'$ be two product distributions over the edges $E$. Assume that all marginals of $\D$ and $\D'$ have supports in $[0, \Delta]$. Then we have
\[|\expost(G, \D) - \expost(G, \D')| \le |V| \cdot \Delta \cdot \sum_{e \in E} d_\mathrm{TV}(D_e, D'_e), \]
and furthermore, for any (one-sided semi-online) decision tree $\mathcal{A}$, we have
\[|\mathcal{A}(G, \D) - \mathcal{A}(G, \D')| \le |V| \cdot \Delta \cdot \sum_{e \in E} d_\mathrm{TV}(D_e, D'_e). \]
\end{corollary}

Equipped with~\Cref{lem:powers-of-eps-make-behaved}, the full construction for~\Cref{thm:qc-equals-semi-online} is straightforward: we discretize every marginal distribution so that the TV distance to the original marginal is not large, while also satisfying the condition of~\Cref{lem:powers-of-eps-make-behaved}. We present the full algorithm and proof for~\Cref{thm:qc-equals-semi-online} in~\Cref{app:proofs}.

\newpage
\bibliographystyle{plainnat}
\bibliography{references}

\newpage
\appendix
\section*{Appendix}

\section{Deferred proofs}\label{app:proofs}
In this chapter of the appendix we present all the proofs that were omitted from the main body.

\subsection{Proof of~\Cref{lemma:exante-compute}}

We begin by providing the definition of the grades $g(\cdot )$ for a Markov chain $\mc$. For any $y\in\R$, consider a penalized game played over $\mc$ where in order to accept a terminal state (and gain its reward) a penalty of $y$ must be paid. Clearly, as $y$ grows to $\infty$, the optimal policy for the penalized game will simply halt at the beginning achieving a utility of $0$. On the other hand, as $y$ approaches $-\infty$, this optimal policy will continue advancing the Markov chain $\mc$ until a terminal state is reached. The grade $g(s)$ of the state $s\in\Omega$ is defined as the minimum value of $y$ for which the optimal policy starting from state $s$ is indifferent towards taking a step in $\mc$ and immediately halting; equivalently, $g(s)$ is the minimum penalty for which immediately halting is optimal.

Up next, we consider a teasing game played over a CICS instance $\inst=(\feas,\mset)$; policies for this game proceed precisely as in CICS; but in order to accept a terminal state from $\mc_i$ that got realized through a trajectory $\traj_i\in\trajset(\mc_i)$ that ends in it, the policy must pay an extra fee of $w_i(\gamma_i)=\min_{s_i\in\gamma_i}g_i(s_i)$. The term teasing refers to the fact that as the trajectory evolves, the minimum grade in the trajectory (and thus our current estimate of the penalty we pay upon accepting) decreases, making the Markov chain more appealing to the policy. We use $\util'_\inst(\pol)$ to denote the utility of a policy $\pol$ in this teasing game. Then, by definition, we have that
\[\util_\inst(\pol) = \util'_\inst(\pol) + \sum_{i=1}^n\sum_{\traj_i\in \trajset(\mc_i)} \probb{\traj_i \text{ realized in }\mc_i \text{ and }\pol\text{ accepts }t(\traj_i)}\cdot w_i(\traj_i)\]
and the proof reduces to proving that $\util'_\inst(\pol)\leq 0$ for all policies $\pol$ with equality if $\pol$ satisfies promise of payment. This is precisely the statement of Lemma A.2 in~\cite{GJSS19} and the proof is ommitted.

\subsection{Proof of~\Cref{lemma:expost-vs-exante}}
Fix any constraint $\feas$ and let $G(\dist):=\expost(\feas,\dist)/\exante(\feas,\dist)$. We will first show that $\corgap(\feas)\geq \min_\dist G(\dist)$. Let $x^*\in\polytope(\feas)$ and $y^*\in\R^n$ be the maximizers for the correlation gap, and consider the Bernoulli product distribution $\dist'=\disti'_1\times\cdots\times\disti'_n$ where $\disti'_i$ takes value $0$ with probability $1-x^*_i$ and value $y^*_i$ otherwise. Since $x^*\in\polytope(\feas)$ and we cannot generate value from $\disti'_i$ with probability larger than $x^*_i$, we have that $\expost(\feas,\dist')=\expectt{\max_{S\in\feas\cap R(x^*)}\sum_{i\in S}y^*_i}$ and $\exante(\feas,\dist')=\sum_{i=1}^nx^*_iy^*_i$. Therefore, we  obtain $\corgap(\feas) = G(\dist') \geq \min_\dist G(\dist)$.

Now, we will prove that $\min_\dist G(\dist)\geq\corgap(\feas)$. Let $\dist^*$ be the maximizer of $G(\dist)$ and let $q^*$ be the selection probabilities in $\exante(\dist^*,\feas)$. Now, consider the Bernoulli product distribution $\dist''=\disti''_1\times\cdots\times\disti''_n$ where $\disti''_i$ takes value $0$ with probability $1-q^*_i$ and value $F_{\disti^*_i}(q^*_i)$ otherwise. By construction, we have that $\exante(\feas,\dist^*)=\exante(\feas,\dist'')$. Furthermore, it is not hard to see that $\expost(\feas,\dist^*)\geq\expost(\feas,\dist'')$; we have simply grouped together all the values in $\disti^*_i$ that lie on the top $q^*_i$ quantile into their conditional expectation, and swapped all other values by $0$. Therefore, we have that $\min_\dist G(\dist) = G(\dist^*) = G(\dist'')\geq \corgap(\feas)$ where the last inequality follows by setting $x=q^*\in\polytope(\feas)$ and $y_i=F_{\disti^*_i}(q^*_i)$ for all $i\in [n]$ in the definition of the correlation gap.

\subsection{Proof of~\Cref{lemma:frugal-vs-semionline}}
Fix any instance $(\feas,\dist)$ of BCS and let $\pol:\R^n\mapsto \feas$ be any $\beta$-frugal algorithm for $\feas$. We will now define a one-sided semi-online algorithm. At the start, we initialize $S=\emptyset$. Then, iteratively:
\begin{enumerate}
    \item For all $i\in [n]$, let $u_i:=\max(\mathrm{supp}(\disti_i))$ be the maximum value in the support of $\disti_i$ and let $u=(u_1,\cdots, u_n)\in\R^n$. Finally, let $i^*$ be the first element not in $S$ that $\pol$ would accept on input $u$; if no such element would be included then we halt.
    
    \item We probe the $i^*$-th element with threshold $u_{i^*}$. If $X_{i^*}=u_{i^*}$ we accept it and set $S=S\cup\{{i^*}\}$. Otherwise, we update our distribution for ${i^*}$ to $\disti_{i^*}=(\disti_{i^*}|X_{i^*}<u_{i^*})$ and repeat step $1$.
\end{enumerate}
Once the above process halts, we return $S$. Let $x=(x_1,\cdots , x_n)$ be any realization from $\dist$. Our objective will be to show that whenever $x$ gets realized, our proposed one-sided semi-online algorithm will accept elements precisely in the same order as $\pol$ does on input $x$; this will immediately provide us with ex-post feasibility as frugal algorithms always accept feasible sets of elements, as well as the approximation result by taking the expectation over $x\sim\dist$ on the frugal guarantee $\sum_{i\in \pol(x)}x_i \geq \beta\cdot\max_{S\in\feas}\sum_{i\in S}x_i$ for all $x\in\R^n$.

Now fix the realization $x$. Observe that by construction, we have that $u_i\geq x_i$ for all $i\in [n]$ throughout the entire execution of the one-sided semi-online algorithm. Let $j$ be the first element that the one-sided semi-online algorithm accepts; again by construction, we know that $u_j=x_j$ on the iteration that $j$ got accepted since we only accept a value when it matches its upper bound. Therefore, by the call to the frugal algorithm during the iteration where $j$ was selected, we know that $j$ is the first element that the frugal algorithm $\pol$ would accept on input $(u_{-j},x_j)$. By the monotonicity of the frugal index function and the fact that $u_i\geq x_i$ for all $i$, this also implies that $j$ is the first element that the frugal algorithm $\pol$ accepts on input $x$.

Thus, we have proven that the first element selected by the one-sided semi-online algorithm will be the same as the first element accepted by $\pol$ on input $x$. Furthermore, since the upper bound for this element will never decrease again after that point and the upper bounds $u_i$ of all other elements will never decrease below their realizations $x_i$, we have that $j$ will be the first element accepted by $\pol$ in all future calls. We can then repeat the exact same argument for the second element that the one-sided semi-online algorithm accepts, then the third, all the way until $\pol$ doesn't accept any more new elements, at which point the one-sided semi-online algorithm halts, proving our initial claim.

\subsection{Proof of~\Cref{lem:eq-prob-make-well-behaved}}
Consider~\Cref{alg:make-behaved}. It is obvious that $\textsc{MakeBehaved}_\mathcal{A}$ is a well-behaved algorithm. We now prove that \[\textsc{MakeBehaved}_\mathcal{A}(G, \mathbf{p}, \mathbf{w}) \ge \mathcal{A}(G, \mathbf{p}, \mathbf{w}).\] We show this by modifying the decision tree of $\mathcal{A}$ under the instance in a series of local steps, the final step of which ends up producing the decision tree of $\textsc{MakeBehaved}_\mathcal{A}$. For a local step, we consider the following modification.
\begin{itemize}
    \item Choose a pair of parallel edges $e, f$ where $w_e > w_f$, and there exists at least one execution path of $\mathcal{A}$ where $f$ is queried before $e$ (or $f$ is queried while $e$ is not). For all action nodes where we query $f$ while having not queried $e$, replace it with querying $e$. For all action nodes where we query $e$ after querying $f$, replace it with querying $f$. Denote this new decision tree to be $\mathcal{A}'$.
\end{itemize}
We first claim that $\mathcal{A}'$ is still a valid decision tree. Note that we did not produce any execution path where we query the same edge twice. In particular, if an execution path queries $e$ before $f$, it stays the same, while if it queries $f$ before $e$, the first query turns into querying $e$ while the second turns into querying $f$. Furthermore, as $e$ and $f$ are parallel edges, the acceptance/rejection at modified action nodes does not affect other nodes for any execution path.

We now show that $\mathcal{A'}(G, \mathbf{p}, \mathbf{w}) \ge \mathcal{A}(G, \mathbf{p}, \mathbf{w})$. As $e$ and $f$ have the same probability of appearance, the only change in performance between the two decision trees lies in the change in gained value when the modified nodes are accepted. In particular, let $N_{e \to f}$ and $N_{f \to e}$ be the sets of action nodes that were changed from querying $e$ to $f$, and from querying $f$ to $e$, respectively. Then,
\begin{align*}
    &\mathcal{A'}(G, \mathbf{p}, \mathbf{w}) - \mathcal{A}(G, \mathbf{p}, \mathbf{w})\\
    &= \sum_{u \in N_{e \to f}} \Pr[\text{$\mathcal{A}$ reaches $u$}] \cdot \Pr[\text{$u$ accepts}] \cdot (w_f - w_e) \\
    &\quad\quad\quad\quad+\sum_{v \in N_{f \to e}} \Pr[\text{$\mathcal{A}$ reaches $v$}] \cdot \Pr[\text{$v$ accepts}] \cdot (w_e - w_f) \\
    &= p(w_e - w_f)\left(\sum_{v \in N_{f \to e}} \Pr[\text{$\mathcal{A}$ reaches $v$}] - \sum_{u \in N_{e \to f}} \Pr[\text{$\mathcal{A}$ reaches $u$}] \right).
\end{align*}
Now, observe that for all $u \in N_{e \to f}$, this (partial) execution path in $\mathcal{A}$ queries $e$ after $f$. Therefore, before reaching $u$, $\mathcal{A}$ must have reached a node $v$ where it queries $f$ while having not queried $e$, or $v \in N_{f \to e}$. This means that
\[\sum_{u \in N_{e \to f}} \Pr[\text{$\mathcal{A}$ reaches $u$}] \le \sum_{v \in N_{f \to e}} \Pr[\text{$\mathcal{A}$ reaches $v$}],\]
and combining with $w_e > w_f$ gives us $\mathcal{A'}(G, \mathbf{p}, \mathbf{w}) - \mathcal{A}(G, \mathbf{p}, \mathbf{w}) \ge 0$ as claimed.

Finally, observe that repeatedly applying the modification step above will end up in producing the decision tree of $\textsc{MakeBehaved}_\mathcal{A}$, while each step improves the performance of the algorithm. Therefore, we have $\textsc{MakeBehaved}_\mathcal{A}(G, \mathbf{p}, \mathbf{w}) \ge \mathcal{A}(G, \mathbf{p}, \mathbf{w})$.

\subsection{Proof of~\Cref{lem:powers-of-eps-make-behaved}}

We first show that any instance $(G , \mathbf{p}, \mathbf{w})$ as above can be reduced to an instance $(G', \mathbf{p'}, \mathbf{w'})$ where every edge $e' \in E'$ has $p_{e'} = \eps$ without affecting the algorithm's performance. Consider~\Cref{alg:eq-prob-reduce}. Note that $\frac{\eps}{1 - (1 - \eps)^{k_e + 1 - i}} \le \frac{\eps}{1 - (1-\eps)^1} = 1$, so the algorithm is well defined.

\begin{algorithm}[tbh]
\caption{$\textsc{EqProbReduce}_\mathcal{A}(G = (V, E), \mathbf{p}, \mathbf{w})$}
\begin{algorithmic}[1]
\label{alg:eq-prob-reduce}
\Require $\mathcal{A}$ is any query-commit algorithm. $(G, \mathbf{p}, \mathbf{w})$ is a query-commit instance where there exists $\eps \in (0, 1)$ such that $p_e = 1 - (1 - \eps)^{k_e}$ for all $e \in E, k_e \in \mathbb{Z}_{\ge 0}$.
\State Let $G' = (V, E' \gets \emptyset), \mathbf{p}', \mathbf{w}'$ be a new query-commit instance.
\For{$e \in E$}
    \State Add $k_e$ copies of $e$ into $E'$ (denoted as $e_i'$ for $i \in [k_e]$), each with probability $p'_{e_i'} = \eps$ and weight $w'_{e_i'} = w_e$.
\EndFor
\State Simulate $\mathcal{A}(G', \mathbf{p}', \mathbf{w}')$.
\Comment{WLOG, assuming $\mathcal{A}$ queries the smallest index copy first.}

\While{$\mathcal{A}$ queries some $e_i'$ on $G'$}
    \State Draw $H_i$ from $\mathrm{Bernoulli}\left(\frac{\eps}{1 - (1 - \eps)^{k_e + 1 - i}}\right)$.
    \If{$H_i = 1$}
        \State Query $e$ on $G$.
        \State Forward the output of the query to $\mathcal{A}$.
    \Else
        \State Forward \textsf{reject} to $\mathcal{A}$.
    \EndIf
\EndWhile
\end{algorithmic}
\end{algorithm}

We prove that $\textsc{EqProbReduce}_\mathcal{A}(G, \mathbf{p}, \mathbf{w}) = \mathcal{A}(G', \mathbf{p}', \mathbf{w}')$. First, $\mathcal{A}$ accepts an edge $e_i'$ in $G'$ iff $\textsc{EqProbReduce}_\mathcal{A}$ accepts an edge $e$ in $G$, both of which connect the same pair of nodes and have the same weight. Furthermore, we show that at any point during the algorithm, the probability of $\mathcal{A}$ accepting an edge $e_i'$, conditioned on $\mathcal{A}$ querying $e_i'$, is exactly $p'_{e_i'} = \eps$.

We show this via induction on the copies of $e \in E$, that being $e_i'$ for $i \in [k_e]$. We overload $k = k_e$. Without loss of generality, assume that $\mathcal{A}$ queries these copies starting from the smallest index first. To simplify the analysis, we also assume that whenever $\mathcal{A}$ rejects $e_i'$, it immediately queries $e'_{i+1}$, since the probability of it querying $e'_{i+1}$ conditioned on rejecting $e_i'$ does not affect the quantity we are calculating.

Define $A_i$ to be the event that $\mathcal{A}$ accepts $e_i'$. We are to prove that for all $i \in [k]$:
\[\Pr\left[A_i \mid \bigcap_{j<i} \bar{A_j}\right] = \eps.\]

When $i = 1$, we have
\[\Pr[A_1] = \Pr[H_1 = 1] \cdot \Pr[\text{edge $e$ exists in $G$}] = \frac{\eps}{1 - (1 - \eps)^k} \cdot (1 - (1 - \eps)^k) = \eps.\]

When $i \ge 2$, note that
\begin{equation}
\label{eq:query-prob}
    \Pr\left[A_i \mid \bigcap_{j<i} \bar{A_j}\right] = \frac{\Pr\left[A_i \cap \bigcap_{j<i} \bar{A_j}\right]}{\Pr\left[\bigcap_{j<i} \bar{A_j}\right]}.\\
\end{equation}
For the numerator of~\Cref{eq:query-prob}, observe that $A_i \cap \bigcap_{j<i} \bar{A_j}$ is true if and only if edge $e$ exists in $G$, $H_i = 1$, and $H_j = 0$ for all $j < i$. Therefore
\begin{align*}
    \Pr\left[A_i \cap \bigcap_{j<i} \bar{A_j}\right] &= \Pr[\text{edge $e$ exists}] \cdot \Pr[H_i = 1] \cdot \prod_{j<i} \Pr[H_j = 0] \\
    &= \Pr[H_i = 1] \cdot \frac{\Pr[H_{i-1}=0]}{\Pr[H_{i-1}=1]} \cdot \Pr[\text{edge $e$ exists}] \cdot \Pr[H_{i-1}=1] \cdot \prod_{j<i-1} \Pr[H_j = 0] \\
    &= \Pr[H_i = 1] \cdot \frac{\Pr[H_{i-1}=0]}{\Pr[H_{i-1}=1]} \cdot \Pr\left[A_{i-1} \cap \bigcap_{j<i-1} \bar{A_j}\right]
\end{align*}
and observe that \[\Pr[H_i = 1] \cdot \frac{\Pr[H_{i-1}=0]}{\Pr[H_{i-1}=1]} = \frac{\eps}{1 - (1 - \eps)^{k+1-i}} \cdot \frac{1 - \frac{\eps}{1 - (1 - \eps)^{k+2-i}}}{\frac{\eps}{1 - (1 - \eps)^{k+2-i}}} = \frac{1 - \eps - (1 - \eps)^{k+2-i}}{1 - (1 - \eps)^{k+1-i}} = 1 - \eps.\]

For the denominator of~\Cref{eq:query-prob}, we have \[\Pr\left[\bigcap_{j<i} \bar{A_j}\right] = \Pr\left[\bar{A}_{i-1} \mid \bigcap_{j<i-1} \bar{A_j}\right] \cdot \Pr\left[\bigcap_{j<i-1} \bar{A_j}\right] = (1 - \eps) \cdot \Pr\left[\bigcap_{j<i-1} \bar{A_j}\right].\]
Therefore,~\Cref{eq:query-prob} turns into
\begin{align*}
    \Pr\left[A_i \mid \bigcap_{j<i} \bar{A_j}\right] = \frac{(1 - \eps) \cdot \Pr\left[A_{i-1} \cap \bigcap_{j<i-1} \bar{A_j}\right]}{(1 - \eps) \cdot \Pr\left[\bigcap_{j<i-1} \bar{A_j}\right]} = \Pr\left[A_{i-1} \mid \bigcap_{j<i-1} \bar{A_j}\right] = \eps.
\end{align*}

We have proven that $\textsc{EqProbReduce}_\mathcal{A}(G, \mathbf{p}, \mathbf{w}) = \mathcal{A}(G', \mathbf{p}', \mathbf{w}')$. Furthermore, since $G'$ is the equiprobability instance satisfying the condition of~\Cref{lem:eq-prob-make-well-behaved}, then if we let $\mathcal{A'} = \textsc{MakeBehaved}_{\mathcal{A}}$, by~\Cref{lem:eq-prob-make-well-behaved} we have \[\textsc{EqProbReduce}_\mathcal{A'}(G, \mathbf{p}, \mathbf{w}) = \mathcal{A'}(G', \mathbf{p}', \mathbf{w}') \ge \mathcal{A}(G', \mathbf{p}', \mathbf{w}').\] Finally, observe that $\textsc{EqProbReduce}_\mathcal{A'}(G, \mathbf{p}, \mathbf{w})$ is well-behaved on $G$. Suppose there are two parallel edges $e, f \in E$ where $w_e > w_f$. Then, as $\mathcal{A'}$ is well-behaved, it must query all copies of $e$ in $G'$ before querying copies $f$. When $\mathcal{A'}$ query the last copy of $e$ in $G'$, i.e. $e'_{k_e}$, $\textsc{EqProbReduce}_\mathcal{A'}$ queries $e$ in $G$ with probability $\frac{\eps}{1 - (1 - \eps)^{k_e + 1 - k_e}} = 1$. This means $\textsc{EqProbReduce}_\mathcal{A'}$ must query $e$ before querying $f$ in $G$, proving its well-behaved property.

\subsection{Proof of~\Cref{thm:qc-equals-semi-online}}
Consider~\Cref{alg:query-commit-reduction}.

\begin{algorithm}[tbh]
\caption{$\textsc{QCReduction}_{\mathcal{A}, \eps}(G = (V, E), \D)$}
\begin{algorithmic}[1]
\label{alg:query-commit-reduction}
\Require $\mathcal{A}$ is a query-commit algorithm. $(G, \D)$ is a one-sided semi-online stochastic matching instance. $\eps \in (0, 1)$ is some parameter.
\State Let $\D'$ be a new empty product measure over $E$.
\For{$e \in E$}
    \State $\sigma_1 \gets 0$. \Comment{$\sigma_i = \Pr_{w \sim D'_e}[w > s_i]$.}
    \State $n \gets |\operatorname{support}(D_e) \cap \R_{> 0}|$.
    \For{$i = 1, 2, \dots, n$}
        \State $s_i \gets$ the $i$-th largest element of $\operatorname{support}(D_e) \cap \R_{> 0}$.
        \State $p_i \gets \Pr_{w \sim D_e}[w = s_i]$.
        \State $x_i \gets \frac{1 - \sigma_i - p_i}{1 - \sigma_i}$. \Comment{Can be shown that $x_i \in [0, 1]$.}
        \State $x_i' \gets \begin{cases} (1 - \eps)^{\lfloor \log_{1 - \eps} (x_i)\rfloor} & \text{if $x_i > 0$} \\ (1 - \eps)^{\lceil \log_{1 - \eps} (\eps)\rceil} & \text{if $x_i = 0$} \end{cases}$. \Comment{$1 \ge x_i' \ge x_i$ as $1 - \eps < 1$.}
        \State $p_i' \gets (1 - \sigma_i)(1 - x_i')$. \Comment{$p_i = (1 - \sigma_i)(1 - x_i)$ so $p_i' \le p_i$.}
        \State Add support $s_i$ into $D'_e$ with probability $p_i'$.
        \State $\sigma_{i+1} \gets \sigma_i + p_i'$.
    \EndFor
    \State Add support $0$ into $D'_e$ with probability $1 - \sigma_{n + 1}$. \Comment{$D'_e$ is now a true probability distribution.}
\EndFor
\State Let $\mathcal{A}' \gets \textsc{SplitSupport}_{\textsc{EqProbReduce}_{\textsc{MakeBehaved}_\mathcal{A}}}$ be a one-sided semi-online algorithm.
\State Simulate $\mathcal{A'}(G, \D')$.
\While{$\mathcal{A'}$ queries some $e \in E$ at threshold $s$}
    \State Query edge $e$ at threshold $s$.
    \State Forward the output of this query to $\mathcal{A'}$.
\EndWhile
\end{algorithmic}
\end{algorithm}
First, we claim that the definition for $D'_e$ is well-defined for any $e \in E$. Let $s_i$ be the $i$-th largest element of $\operatorname{support}(D_e) \cap \R_{> 0}$. It is obvious that $\sigma_i = \Pr_{w \sim D'_e}[w > s_i]$ and $p_i' = \Pr_{w \sim D'_e}[w = s_i]$ at every iteration. We aim to prove that $p_i' \le p_i$ at every iteration, which implies $\sigma_{i+1} = \Pr_{w \sim D'_e}[w > s_{i+1}] \le \Pr_{w \sim D_e}[w > s_{i+1}]$ and at every iteration. We do this by inducting on the iteration $i$, noting that $\sigma_1$ satisfies the condition $\sigma_{1} = 0 = \Pr_{w \sim D_e}[w > s_1]$.

First of all, observe that
\[\frac{1 - \sigma_i - p_i}{1 - \sigma_i} = 1 - \frac{p_i}{1 - \sigma_i} \ge 1 - \frac{\Pr_{w \sim D_e}[w = s_i]}{1 - \Pr_{w \sim D_e}[w > s_i]} = 1 - \frac{\Pr_{w \sim D_e}[w = s_i]}{\Pr_{w \sim D_e}[w \le s_i]} \ge 0,\] where the first inequality uses the inductive hypothesis. Therefore, at line 8, we always have $x_i \in [0, 1]$. Therefore, $x_i'$ is well-defined, and it obviously holds that $x_i' \in [x_i, 1]$. Therefore, since $p_i = (1 - \sigma_i)(1 - x_i)$ by line 8, and $p_i' = (1 - \sigma_i)(1 - x_i')$ by line 10, we have $p_i \le p_i'$ as claimed. This means that at line 14, we have $\sigma_{n+1} \in [0, 1]$, so line 14 is well-defined, while also turning $D'_e$ into a true distribution.

At line 16 and 17, we use $\textsc{SplitSupport}$ to turn the one-sided semi-online stochastic matching instance $G, \D'$ into a query-commit instance $G', \mathbf{p}', \mathbf{w}'$, and then we apply~\Cref{lem:powers-of-eps-make-behaved} on this query-commit instance (by further reducing the instance via $\textsc{EqProbReduce}$ while turning $\mathcal{A}$ well-behaved via $\textsc{MakeBehaved}$). We prove that the query-commit instance $G', \mathbf{p}', \mathbf{w}'$ produced via $\textsc{SplitSupport}$ (\Cref{alg:split-support}) satisfies the condition of~\Cref{lem:powers-of-eps-make-behaved}. Recall that very copy $e'_{s_i}$ in $G'$ has its probability equals
\[1 - \frac{\Pr_{w \sim D'_e}[w < s_i]}{\Pr_{w \sim D'_e}[w \le s_i]} = 1 - \frac{1 - \Pr_{w \sim D'_e}[w > s_i] - \Pr_{w \sim D'_e}[w = s_i]}{1 - \Pr_{w \sim D'_e}[w > s_i]}\]
and note that this probability equals
\[1 - \frac{1 - \sigma_i - p_i'}{1 - \sigma_i} = 1 - \frac{1 - \sigma_i - (1 - \sigma_i)(1 - x_i')}{1 - \sigma_i} = 1 - x_i'\]
by line 10 of~\Cref{alg:query-commit-reduction}. We have $1 - x_i' = 1 - (1 - \eps)^{k}$ for some integer $k \in \mathbb{Z}_{\ge 0}$ from line 9, so the instance $G'$ satisfying the condition of~\Cref{lem:powers-of-eps-make-behaved}. Therefore, $\mathcal{A'}$ defined in line 16 is a well-defined one-sided semi-online algorithm.

Finally, we prove that for every $e \in E$, we have $d_\mathrm{TV}(D_e, D'_e) \le \eps |\operatorname{support}(\mathcal{D}_e)|$. To this end, we show that at every iteration, we have $p_i - \eps \le p_i' \le p_i$. We have two cases.
\begin{description}
    \item[Case $x_i = 0$.] Then, we have $p_i = 1 - \sigma_i$ and $0 \le x_i' \le \eps$ by line 8 and 9, so $p_i' = (1 - \sigma_i)(1 - x_i') \ge p_i(1 - \eps) \ge p_i - \eps$ and $p_i' = (1 - \sigma_i)(1 - x_i') \le p_i$.
    \item[Case $x_i > 0$.] Then, we have $\frac{x_i}{1 - \eps} > x_i' \ge x_i$ by line 9. Since $x_i + \eps \ge \frac{x_i}{1 - \eps}$, we have $p_i' = (1 - \sigma_i)(1 - x_i') \ge (1 - \sigma_i)(1 - x_i - \eps) \ge p_i - \eps$ and $p_i' = (1 - \sigma_i)(1 - x_i') \le (1 - \sigma_i)(1 - x_i) = p_i$.
\end{description}
Let $\Delta$ be the maximum support of any edge, $U = |V| \cdot \Delta \cdot \sum_{e \in E} |\operatorname{support}(\mathcal{D}_e)|$, $\alpha$ be the competitive ratio of $\mathcal{A}$, and $(G', \mathbf{p}', \mathbf{w}')$ be the query-commit instance after applying~\Cref{alg:split-support} on $\D'$. By~\Cref{cor:tv-distance} and \Cref{lem:powers-of-eps-make-behaved,lem:well-behaved-qc-equals-semi-online},
\[
    \expost(G, \D) \le \expost(G, \D') + U \eps  = \expost(G', \mathbf{p}', \mathbf{w}') + U \eps \le \frac{1}{\alpha} \cdot \mathcal{A}(G', \mathbf{p}', \mathbf{w}') + U \eps
\]
and
\[
    \textsc{QCReduction}_{\mathcal{A}, \eps}(G, \D) \ge \mathcal{A'}(G, \D') - U \eps \ge \mathcal{A}(G', \mathbf{p}', \mathbf{w}') - U \eps
\]
so
\[\textsc{QCReduction}_{\mathcal{A}, \eps}(G, \D) \ge \alpha \cdot \expost(G, \D) - (\alpha + 1)U\eps.\]
For any $\eps' > 0$, setting $\eps = \frac{\mathcal{A}(G', \mathbf{p}', \mathbf{w}') \eps'}{U(1 + \alpha^{-1} + \eps')}$ gives us $\frac{\textsc{QCReduction}_{\mathcal{A}, \eps}(G, \D)}{\expost(G, \D)} \ge \alpha - \eps'$.



\section{Commitment Gap for Minimization}\label{app:minimization}

In this section, we discuss the minimization version of CICS (henceforth, min-CICS) where a policy pays for the cost of the actions it takes \textit{plus} the value of the terminal states it accepts. We note that by appropriately adjusting our definitions from maximization to minimization, our entire framework of relating different benchmarks would seamlessly extend to the minimization setting. 

However, this approach would be hopeless; the correlation gap (i.e. the worst case ratio between the costs of the ex post and ex ante optima) is potentially unbounded even under single-selection constraints and therefore the same is true for the frugal correlation gap. As it turns out, this is not a weakness of our framework but rather an inherent issue with min-CICS; the commitment gap of min-CICS is also unbounded. We illustrate such examples in~\Cref{thm:minim}. 

Finally, we note that this doesn't contradict the results of~\cite{CCHS24} who bound the commitment gap for specific variants of min-CICS; it simply states that (in contrast to the maximization setting) a universal bound on the commitment gap of all min-CICS instances, even under single selection constraints, is unattainable.

\begin{theorem}\label{thm:minim}
    There are min-CICS instances $\inst=(\feas,\mset)$ comprised by two MDPs $\mset=(\mc_1,\mc_2)$ and a single-selection constraint $\feas=\{\{1\},\{2\},\{1,2\}\}$ such that the commitment gap is unbounded.
\end{theorem}

\begin{proof} Fix any parameter $N>1$. We consider a min-CICS instance $\inst=(\feas,\mset)$ comprised by two MDPs $\mset=(\mc_1,\mc_2)$ and a single-selection constraint $\feas=\{\{1\},\{2\},\{1,2\}\}$. The two MDPs are depicted in~\Cref{fig:min-cics-example}. The first one is a Markov chain with a unique action of cost $0$ that transitions to terminal states of value $\infty$ and $1$ with probabilities $p$ and $1-p$ respectively, where $p=1/(N+1)$. The second MDP has two available $0$-cost actions $\{a,b\}$ at its root state. The first deterministically leads to a terminal state of value $1$ and the second leads to terminal states of value $2N^2-1$ and $1/(2N-1)$ with probabilities $q$ and $1-q$ respectively, where $q=1/(2N)$.

\begin{figure}[ht!]
    \centering
    \resizebox{0.9\textwidth}{!}{%
    \begin{tikzpicture}[
        node distance=1.1.5cm,
        circ/.style={circle, draw, minimum size=8mm},
        dot/.style = {circle, fill, minimum size=#1,
                  inner sep=0pt, outer sep=0pt},
        arrow/.style={-Latex, thick},
        dotline/.style={dotted, thick}
    ]
    
    \begin{scope}[xshift=-4cm] 
        \node[circ, label={left:$\mathcal{M}_1$}] (M1_top) {};
        \node[dot, below=1cm of M1_top] (M1_mid) {};
        \node[circ, below left=1.5cm and 1.5cm of M1_mid, label={below:$v=1$}] (M1_bottom_left) {};
        \node[circ, below right=1.5cm and 1.5cm of M1_mid, label={below:$v=\infty$}] (M1_bottom_right) {};
    
        \draw[arrow] (M1_top) -- (M1_mid) {}; 
        \draw (M1_mid) node[left=1mm,font=\footnotesize] {$1-p$};
        \draw (M1_mid) node[right=1mm,font=\footnotesize] {$p = \frac{1}{N+1}$};
    
        \draw[dotline] (M1_mid) -- (M1_bottom_left);
        \draw[dotline] (M1_mid) -- (M1_bottom_right);
    \end{scope}

    \begin{scope}[xshift=4cm] 
        \node[circ, label={left:$\mathcal{M}_2$}] (M2_top) {};
        \node[circ, below left=1cm and 1cm of M2_top,label={below:$v=1$}] (M2_mid_left) {};
        \node[dot, below right=1cm and 1cm of M2_top] (M2_mid_right) {}; 
        \node[circ, below left=1.5cm and 1.5cm of M2_mid_right,label={below:$v = (2N-1)^{-1}$}] (M2_bottom_left) {}; 
        \node[circ, below right=1.5cm and 1.5cm of M2_mid_right,label={below:$v = 2N^2-1$}] (M2_bottom_right) {}; 

        \draw[arrow] (M2_top) -- (M2_mid_left) node[midway, above left,font=\footnotesize] {$a$};
        \draw[arrow] (M2_top) -- (M2_mid_right) node[midway, above right,font=\footnotesize] {$b$};
    
        \draw[dotline] (M2_mid_right) -- ($(M2_mid_right)!0.5!(M2_bottom_left)$) coordinate (M2_mid_left_branch);
        \draw (M2_mid_right) node[left=1mm,font=\footnotesize] {$1-q$};
        \draw[dotline] (M2_mid_left_branch) -- (M2_bottom_left);
    
        \draw[dotline] (M2_mid_right) -- ($(M2_mid_right)!0.5!(M2_bottom_right)$) coordinate (M2_mid_right_branch);
        \draw (M2_mid_right) node[right=1mm,font=\footnotesize] {$q = \frac{1}{2N}$};
        \draw[dotline] (M2_mid_right_branch) -- (M2_bottom_right);
    \end{scope}
    
    \end{tikzpicture}
    }%
    \caption{A min-CICS instance for which the commitment gap approaches zero as $N\rightarrow\infty$. Black arrows denote costly actions (although all actions have cost $0$ in this example) and dotted lines denote random transitions.}
    \label{fig:min-cics-example}
\end{figure}
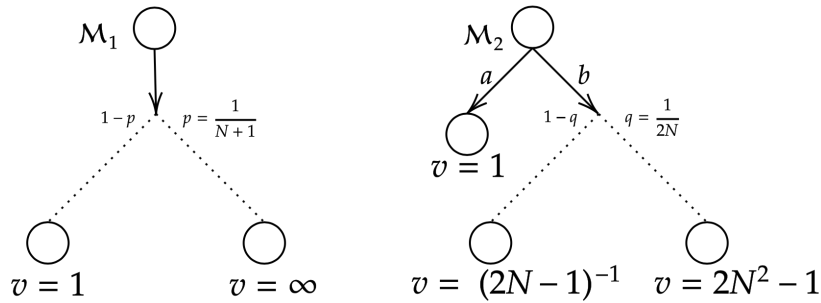

We will first compute the optimal policy for $\inst$. Without loss (since all action costs are $0$), this policy begins by advancing $\mc_1$; if the realized terminal has value $1$ then it will prefer action $b$ in $\mc_2$ to try and realize an even lower value; and if the realized terminal in $\mc_1$ has value $\infty$ then it will prefer action $a$ to secure a value of $1$. Therefore, we have that
\[\opt(\inst)= (1-p)\cdot (1-q)\cdot\frac{1}{2N-1} + (1-p)\cdot q\cdot 1 + p\cdot 1 = \frac{2}{N+1}.\]

We now proceed to compute the optimal committing policy for $\inst$. Note that there are only two possible commitments, based on whether action $a$ or $b$ in $\mc_2$ is preferred. If the policy commits to $a$, then the minimum value across realized terminal states is clearly $1$. If the policy commits to $b$, then the expected minimum value across terminal states is
\[(1-q)\cdot\frac{1}{2N-1} + q\cdot (1-p)\cdot 1 + q\cdot p\cdot (2N^2-1)=1\]
and therefore both commitments have an expected cost of $1$; the same will be true for any randomized commitment that takes action $a$ with some probability $\lambda$ and action $b$ with probability $1-\lambda$. Therefore, we see that the ratio between the cost of any committing policy and the cost of the overall optimal policy (i.e. the commitment gap) is $(N+1)/2$ which can be made arbitrarily large as $N\rightarrow \infty$.

We also note that by the constructive nature of our reductions, this hard example can be transformed into a min-BCS example where the cost of the optimal semi-online algorithm is arbitrarily larger than the cost of the ex-ante optimal benchmark. In particular, solving for the ex ante optimal policy of instance $\inst$, one can easily see that it will commit to action $b$ of $\mc_2$ and accept the small value $v=(2N-1)^{-1}$ whenever it gets realized (so with probability $1-1/(2N)$) while also accepting $v=1$ from $\mc_1$ with probability $1/(2N)\leq 1-p$ to satisfy the ex ante constraint. Therefore, the optimal ex ante feasible policy for $\inst$ is described as follows:
\begin{enumerate}
    \item Advance $\mc_1$. If $v=1$ accept it with probability $\frac{N+1}{2N^2}$. If $v=\infty$ reject it.
    \item Advance $\mc_2$ through action $b$. If $v = \frac{1}{2N-1}$ then accept it. If $v = 2N^2-1$ reject it.
\end{enumerate}
The Markov chains that this commitment induces only have zero cost actions, and therefore the underlying distributions of surrogate costs simply correspond to the value of a randomly realized terminal. Therefore, we obtain the instance $(\feas,X_1\times X_2)$ of min-BCS where
\[X_1 = 
\begin{cases}
            1 & \text{w.p. }\; 1-\frac{1}{N+1}\\
           \infty & \text{w.p. }\; \frac{1}{N+1}\\
\end{cases} \quad , \quad 
X_2 = \begin{cases}
            \frac{1}{2N-1} & \text{w.p. }\; 1-\frac{1}{2N}\\
            2N^2-1 & \text{w.p. }\; \frac{1}{2N}\\
\end{cases}  .
\]
Indeed, the cost of the optimal semi-online algorithm (which equals the cost of the ex post optimum $\expectt{\min\{X_1,X_2\}}$) in the above instance is $1$, whereas the cost of the ex ante optimal is $N^{-1}$, showing that the correlation gap is unbounded.
\end{proof}

\end{document}